%% file: main.tex
\documentclass[runningheads,a4paper]{llncs}

\input{styles.tex}
\input{macro.tex}

\newtoggle{anonymous}
\togglefalse{anonymous}
\newtoggle{arxiv}
\toggletrue{arxiv}

\makeatletter
\renewcommand\paragraph{\@startsection{paragraph}{4}{\z@}%
  {2.25ex \@plus 1ex \@minus .2ex}%
  {-0.75em}%
  {\normalfont\normalsize\bfseries}}
\makeatother

\begin{document}

\title{A formal model of Algorand smart contracts}

\iftoggle{anonymous}{
  \author{\vspace{-15mm}}
  \institute{}
  \authorrunning{Anonymous Author(s)}
}{
\author{Massimo Bartoletti\inst{1}, Andrea Bracciali\inst{2}, Cristian Lepore\inst{2}, \\ Alceste Scalas\inst{3}, Roberto Zunino\inst{4}}
\authorrunning{Bartoletti et al.}

\institute{
Universit\`a degli Studi di Cagliari, Cagliari, Italy
\and
Stirling University, Stirling, UK
\and
Technical University of Denmark, Lyngby, Denmark
\and
Universit\`a degli Studi di Trento, Trento, Italy}
}

\Crefname{figure}{Fig.\@}{Figures}%
\Crefname{section}{\S\!\!\@}{Sections}%
\Crefname{appendix}{\S\!\!\@}{Appendices}%

\maketitle

\input{abstract.tex}

\input{intro.tex}

\input{model.tex}

\input{examples.tex}

\input{tool.tex}

\input{conclusions.tex}

\input{ack.tex}

\bibliographystyle{splncs04}
\bibliography{main}

\iftoggle{arxiv}{
  \newpage
  \appendix
  \input{app-model.tex}
  \input{app-examples.tex}
  \input{app-tool.tex}

  \input{proofs.tex}
}
{}

\end{document}

%% file: styles.tex
\usepackage[utf8]{inputenc}
\usepackage{color}
\usepackage[usenames,dvipsnames]{xcolor}

\usepackage{centernot}
\usepackage{hhline}

\usepackage[mathcal]{euscript}

\usepackage{siunitx}
\usepackage{textcomp}
\usepackage{caption}
\usepackage{pst-func}
\usepackage{amssymb}
\usepackage{graphicx}
\usepackage{paralist}
\usepackage{environ}
\usepackage{pgfplots}
\usepackage{tikz}
\usetikzlibrary{shadows}
\usepackage{mathtools}
\usepackage{array,multirow}
\usepackage[hyphens]{url}
\usepackage{cite}
\usepackage{pgfplotstable}
\usepackage{threeparttable}

\usepackage{etoolbox}
\usepackage{arydshln}

\usepackage{xspace} 
 
\usepackage{dirtytalk} 
\usepackage{stmaryrd} 

\usepackage[inline,shortlabels]{enumitem} 
\newlist{inlinelist}{enumerate*}{1}
\setlist*[inlinelist,1]{%
  label=(\roman*),
}

\usepackage{xifthen}  
\usepackage{ifthen}

\usepackage{nicefrac}

\usepackage[hidelinks]{hyperref}
\usepackage[hyphenbreaks]{breakurl}
\usepackage{cleveref}
\usepackage{bigints}

\usetikzlibrary{pgfplots.dateplot}
\usetikzlibrary{matrix,chains,positioning,decorations.pathreplacing}
\usetikzlibrary{patterns,calc,fit,arrows}
\usetikzlibrary{shapes.geometric}

\usepackage[final,margin,index]{fixme} 
\fxusetheme{color}

\FXRegisterAuthor{bart}{anbart}{\color{magenta} {\underline{bart}}}
\FXRegisterAuthor{zun}{anzun}{\color{red} {\underline{zun}}}
\FXRegisterAuthor{alc}{analc}{\color{blue} {\underline{alc}}}
\FXRegisterAuthor{and}{anand}{\color{ForestGreen} {\underline{abb}}}

\hypersetup{
  breaklinks   = true,
  colorlinks   = true, 
  urlcolor     = blue, 
  linkcolor    = blue, 
  citecolor    = red   
}
\hypersetup{final} 

\usepackage{listings}
\definecolor{listingBG}{HTML}{FFFFCB}%
\definecolor{listingFrame}{HTML}{BBBB98}%
\definecolor{listingLineno}{rgb}{0.5,0.5,1.0}%

\definecolor{LightGrey}{rgb}{0.975,0.975,0.975}

\lstdefinelanguage{asc}{
	commentstyle=\color{Gray},
	morecomment=[l]{//},
	morecomment=[s]{/*}{*/},
	classoffset=0,
        escapechar=\$,
	morekeywords={type,close,snd,rcv,amt,fee,fv,lv,asnd,arcv,aclose,xaid,aamt},
	keywordstyle=\color{Blue}\bfseries,
	classoffset=1,
	morekeywords={sig,versig,fun,unit,int,string,bool,address,uint},
	keywordstyle=\color{TealBlue},
	classoffset=2,
	morekeywords={BTC,true},
	keywordstyle=\color{Plum}\bfseries,
}

\lstdefinelanguage{btm}{
	commentstyle=\color{Gray},
	morecomment=[l]{//},
	morecomment=[s]{/*}{*/},
	classoffset=0,
        escapechar=\$,
	morekeywords={transaction,input,output,key,network,package},
	keywordstyle=\color{Blue}\bfseries,
	classoffset=1,
	morekeywords={sig,versig,fun,unit,int,string,bool,address,uint},
	keywordstyle=\color{TealBlue},
	classoffset=2,
	morekeywords={BTC,true},
	keywordstyle=\color{Plum}\bfseries,
}

\lstdefinelanguage{solidity}{
	commentstyle=\color{Gray},
	morecomment=[l]{//},
	morecomment=[s]{/*}{*/},
	classoffset=0,
        escapechar=\$,
	morekeywords={struct,mapping,function,this,public,private,static,final,class,extends,switch,case,break,finally,try,catch,return,if,else,new},
	keywordstyle=\color{Blue}\bfseries,
	classoffset=1,
	morekeywords={unit,int,string,bool,address,uint},
	keywordstyle=\color{TealBlue},
	classoffset=2,
	morekeywords={ether,wei,finney,contract,send,throw,msg,sender,value},
	keywordstyle=\color{Plum}\bfseries,
}

\lstdefinelanguage{java}{
	escapechar=\$,
        commentstyle=\color{Gray},
	morecomment=[l]{//},
	morecomment=[s]{/*}{*/},
	morestring=[b]",
        classoffset=0,
	morekeywords={public,private,static,final,class,extends,switch,case,break,finally,try,catch,void,int,boolean,throws,throw,return,if,else,new},
	keywordstyle=\color{keyword}\bfseries
}

%% file: macro.tex
\newcommand{\ifempty}[3]{%
  \ifthenelse{\isempty{#1}}{#2}{#3}%
}

\newcommand{\hidden}[1]{}

\newcommand{\eqdef}{\mathrel{\triangleq}}

\newcommand{\Real}[1]{\mathrm{Real}}

\newcommand{\codefont}{\fontsize{10}{10}\selectfont}
\newcommand{\code}[1]{{\tt\codefont {#1}}}

\newcommand{\eg}{e.g.\@\xspace}
\newcommand{\ie}{i.e.\@\xspace}

\newcommand{\keyterm}[1]{\textbf{\emph{#1}}}%

\newcommand{\emptyseq}{\varepsilon}

\newcommand{\fst}{\mathit{fst}}
\newcommand{\snd}{\mathit{snd}}

\newcommand{\den}[2][]{#2_{#1}} 
\newcommand{\denV}[1][]{\den[#1]{\nu}}

\newenvironment{proofof}[2][]{%
  \ifempty{#1}
  {\subsection*{Proof of~\Cref{#2}}}
  {\subsection*{Proof of~\Cref{#2} ({#1})}}
  \label{#2-proof}
  }%
  {}

\newcommand{\sig}[3][]{\mathit{sig}^{#1}_{#2}\ifempty{#3}{}{({#3})}}

\newcommand{\PK}[1]{k^p_{#1}}
\newcommand{\SK}[1]{k^s_{#1}}

\newcommand{\secret}[2]{s_{#1}^{#2}}
\newcommand{\hash}[2]{h_{#1}^{#2}}

\newcommand{\BTC}{\textup{%
  \leavevmode
  \vtop{\offinterlineskip 
    \setbox0=\hbox{B}%
    \setbox2=\hbox to\wd0{\hfil\hskip-.03em
    \vrule height .3ex width .15ex\hskip .08em
    \vrule height .3ex width .15ex\hfil}
    \vbox{\copy2\box0}\box2}}}

\DeclareMathAlphabet{\mathbfsf}{\encodingdefault}{\sfdefault}{bx}{n}

\newcommand{\secteal}{\texttt{secteal}\xspace}

\def\pmvColor{\color{ForestGreen}}
\newcommand{\pmvFmt}[1]{{\pmvColor{\sf #1}}}
\newcommand{\pmv}[2][]{\pmvFmt{#2}_{\pmvColor{#1}}\xspace}
\newcommand{\pmvA}[1][]{\pmv[{#1}]{a}}
\newcommand{\pmvB}[1][]{\pmv[{#1}]{b}}
\newcommand{\pmvC}[1][]{\pmv[{#1}]{c}}

\newcommand{\pmvAS}[1][]{\pmv[{#1}]{A}}

\newcommand{\pmvAL}[1][]{\pmv[{#1}]{\mathcal{A}}}

\newcommand{\PmvU}{{\pmvColor{\mathbb{K}}}}
\newcommand{\Adv}{\pmv{Adv}} 

\newcommand{\runnameS}{\mathcal{R}}

\newcommand{\runS}[1][]{\runnameS_{#1}}

\newcommand{\stratS}[1]{\Sigma_{#1}}

\newcommand{\stratSet}{\mathbf{\Sigma}} 

\newcommand{\val}[2][]{#2_{#1}} 
\newcommand{\valV}[1][]{\val[#1]{v}}

\newcommand{\wit}[2][]{#2_{#1}} 
\newcommand{\witW}[1][]{\wit[#1]{w}}
\newcommand{\witWi}[1][]{\wit[#1]{w'}}
\newcommand{\witWL}[1][]{\wit[#1]{\mathcal{W}}}
\newcommand{\witWLi}[1][]{\wit[#1]{\mathcal{W}'}}
\newcommand{\witWLL}[1][]{\wit[#1]{\mathbf{W}}}
\newcommand{\witWLLi}[1][]{\wit[#1]{\mathbf{W}'}}

\def\timeColor{\color{black}}

\newcommand{\timeT}[1][]{\timeColor{r_{#1}}}
\newcommand{\timeTi}[1][]{\timeColor{r'_{#1}}}
\newcommand{\timeTmax}[1][]{\timeColor{r_{\it max}}}
\newcommand{\timeTinit}[1][]{\timeColor{r_{\it init}}}
\newcommand{\MaxTxnLife}{\timeColor{\Delta}_{max}}

\newcommand{\leaseMap}{f_{\textit{lx}}}
\newcommand{\leaseMapi}{f'_{\textit{lx}}}

\newcommand{\freezeMap}{f_{\textit{frz}}}
\newcommand{\freezeMapi}{f'_{\textit{frz}}}

\newcommand{\leaseUpdate}[3]{\textit{upd}({#1},{#2})}

\newcommand{\timeOk}[2]{{#1} \models {#2}}

\def\actColor{\color{magenta}}
\newcommand{\actFmt}[1]{{\actColor{#1}}}
\newcommand{\act}[2][]{\actFmt{#2}_{\actColor{#1}}\xspace}
\newcommand{\actL}[1][]{\act[{#1}]{t}}
\newcommand{\actLi}[1][]{\act[{#1}]{t'}}

\newcommand{\actf}{\actFmt{\sf f}} 

\newcommand{\actLS}[1][]{\act[{#1}]{T}}
\newcommand{\actLiS}[1][]{\act[{#1}]{T'}}
\newcommand{\actLSdec}[2][]{{\actFmt{T}_{#1}^{\actColor{\it #2}}}} 

\newcommand{\actLL}[1][]{\act[{#1}]{\mathcal{T}}}
\newcommand{\actLLi}[1][]{\act[{#1}]{\mathcal{T}'}}

\newcommand{\ActU}{{\actFmt{\mathbb{T}}}}

\newcommand{\actPay}{\actFmt{\it pay}}
\newcommand{\actClose}{\actFmt{\it close}}
\newcommand{\actGen}{\actFmt{\it gen}}
\newcommand{\actOptin}{\actFmt{\it optin}}
\newcommand{\actBurn}{\actFmt{\it burn}}
\newcommand{\actFreeze}{\actFmt{\it frz}}
\newcommand{\actUnfreeze}{\actFmt{\it unfrz}}
\newcommand{\actRevoke}{\actFmt{\it rvk}}
\newcommand{\actDelegate}{\actFmt{\it delegate}}

\newcommand{\actType}{{\actColor{\sf type}}}
\newcommand{\actSnd}{{\actColor{\sf snd}}}
\newcommand{\actRcv}{{\actColor{\sf rcv}}}
\newcommand{\actTok}{{\actColor{\sf asst}}}
\newcommand{\actAmt}{{\actColor{\sf val}}}
\newcommand{\actTf}{{\actColor{\sf fv}}}
\newcommand{\actTl}{{\actColor{\sf lv}}}
\newcommand{\actLease}{{\actColor{\sf lx}}}

\newcommand{\actDupSet}{{\actColor{T_{\actTl}}}}
\newcommand{\actDupSeti}{{\actColor{T'_{\actTl}}}}

\newcommand{\labS}[1][]{\act[{#1}]{\ell}}

\newcommand{\labTime}{\timeColor{\checkmark}} 

\newcommand{\auth}[2]{{#1}{:}{#2}}
\newcommand{\authver}[2]{{#1} \models {#2}}

\def\tokColor{\color{MidnightBlue}}
\newcommand{\tokFmt}[1]{{\tokColor{#1}}}
\newcommand{\tok}[2][]{\tokFmt{#2}_{\tokColor{#1}}\xspace}
\newcommand{\tokT}[1][]{\tok[{#1}]{\tau}}
\newcommand{\tokTi}[1][]{\tok[{#1}]{\tau'}}

\newcommand{\tokMngr}[2][]{{\tokColor{\it auth}}\ifempty{#1}{}{({#2},{#1})}}
\newcommand{\mngrMap}{f_{\textit{asst}}}

\newcommand{\ALGO}{\mbox{\tokColor{\textup{\textsf{Algo}}}}}

\newcommand{\TokU}{{\tokColor{\mathbb{A}}}}

\newcommand{\tokval}[3][]{\mathit{val}^{#1}_{{#2}}\ifempty{#3}{}{({#3})}}

\def\addrColor{\color{ForestGreen}}
\newcommand{\addrFmt}[1]{{\addrColor{\mathbf{#1}}}}
\newcommand{\addr}[2][]{\addrFmt{#2}_{\addrColor{#1}}\xspace}
\newcommand{\addrX}[1][]{\addr[{#1}]{x}}

\newcommand{\addrY}[1][]{\addr[{#1}]{y}}

\newcommand{\addrInit}{\pmvA[0]}

\newcommand{\AddrU}{{\addrColor{\mathbb{X}}}}

\newcommand{\tokMapS}[1][]{{\tokColor{\sigma_{#1}}}}
\newcommand{\tokMapSi}[1][]{{\tokColor{\sigma'_{#1}}}}

\newcommand{\acct}[2]{{#1}{\tokColor{[{#2}]}}}

\newcommand{\acctOk}[1]{\models {#1}}

\def\confColor{\color{ForestGreen}}
\newcommand{\conf}[2][]{{\confColor{{#2}_{#1}}}}
\newcommand{\confG}[1][]{\conf[{#1}]{\Gamma}}
\newcommand{\confGi}[1][]{\conf[{#1}]{\Gamma'}}
\newcommand{\confGii}[1][]{\conf[{#1}]{\Gamma''}}
\newcommand{\confK}[1][]{\conf[{#1}]{K}}
\newcommand{\confKi}[1][]{\conf[{#1}]{K'}}
\newcommand{\confKii}[1][]{\conf[{#1}]{K''}}

\newcommand{\confNil}{\conf{0}}

\newcommand{\step}[1]{\xrightarrow{#1}_1}
\newcommand{\stepL}[1]{\xrightarrow{#1}}

\newcommand{\stepTL}[1]{\xRightarrow{#1}}

\def\scriptColor{\color{Plum}}
\newcommand{\script}[2][]{{\scriptColor{{\it #2}_{#1}}}}
\newcommand{\scriptE}[1][]{\script[{#1}]{e}}
\newcommand{\scriptEi}[1][]{\script[{#1}]{e'}}

\newcommand{\ScriptU}{{\scriptColor{\mathbb{S}}}}

\newcommand{\actE}[2][]{\script{\sf tx}\ifempty{#1}{}{({#1})}.{#2}}
\newcommand{\argE}[1]{\script{\sf arg}\ifempty{#1}{}{({#1})}}
\newcommand{\actLen}{\script{\sf txlen}}
\newcommand{\actPos}{\script{\sf txpos}}
\newcommand{\actId}[1][]{\script{\sf txid}\ifempty{#1}{}{({#1})}}
\newcommand{\versigname}{\script{\sf versig}}
\newcommand{\versig}[3]{{\versigname}({#1},{#2},{#3})}

\newcommand{\hashE}[1]{\script{\sf H}(#1)}

\newcommand{\op}[3]{\ensuremath{#2\mathbin\textsf{#1}#3}}
\newcommand{\notE}{{\sf not}~}

\newcommand{\ifE}[3]{\script{\mathsf{if}}~{#1}~\script{\mathsf{then}}~{#2}~\script{\mathsf{else}}~{#3}}

\newcommand{\andE}{~{\sf and}~}
\newcommand{\orE}{~{\sf or}~}

\newcommand{\true}{\mathit{true}}
\newcommand{\false}{\mathit{false}}

\newcommand{\irule}[2]{\dfrac{#1}{#2}}

\newcommand{\sem}[3]{\mbox{\ensuremath{\llbracket{#3}\rrbracket_{#1}^{#2}}}}

\newcommand{\hashSem}[1]{\ifempty{#1}{H}{H(#1)}}

\newcommand{\versigSem}[3]{{\mathit{ver}}_{#1}({#2},{#3})}

\newcommand{\mapstopart}{\rightharpoonup}

\newcommand{\bnfdef}{\Coloneqq}
\newcommand{\bnfmid}{\;|\;}

\newcommand{\nrule}[1]{{\scriptsize \textsc{#1}}}

\newcommand{\dom}[1]{\operatorname{dom} {#1}}

\newcommand{\Nat}{\mathbb{N}}
\newcommand{\Int}{\mathbb{Z}}
\newcommand{\UIntX}{\mathbb{U}_{64}}

\newcommand{\powset}[1]{2^{#1}}
\newcommand{\bind}[2]{{#1} \!\mapsto\! {#2}}
\newcommand{\setenum}[1]{\{#1\}}
\newcommand{\setcomp}[2]{\left\{{#1} \,\middle|\, {#2}\right\}}

\newcommand{\seq}[1]{\langle {#1} \rangle}
\newcommand{\seqat}[2]{{#1}.{#2}}

\newcommand{\setofseq}[1]{\mathit{set}({#1})}

\spnewtheorem{notation}{Notation}{\bfseries}{\rmfamily}

\definecolor{LightGrey}{rgb}{0.95,0.95,0.95}
\definecolor{keyword}{HTML}{7F0055}

\newlength\replength
\newcommand\repfrac{.1}

\newcommand\rulewidth{.6pt}
\newcommand\tdashfill[1][\repfrac]{\cleaders\hbox to \replength{%
  \smash{\rule[\arraystretch\ht\strutbox]{\repfrac\replength}{\rulewidth}}}\hfill}

\newcommand\tdotfill[1][\repfrac]{\cleaders\hbox to \replength{%
  \smash{\raisebox{\arraystretch\dimexpr\ht\strutbox-.1ex\relax}{.}}}\hfill}



%% file: abstract.tex
\begin{abstract}
  We develop a formal model of Algorand stateless smart contracts
  (stateless ASC1).
  We exploit our model to prove fundamental properties of the Algorand blockchain, and to establish the security of some archetypal smart contracts. 
  While doing this, we highlight various design patterns supported by Algorand.
  We perform experiments to validate the coherence of our formal model w.r.t.~the actual implementation.
\end{abstract}

%% file: intro.tex
\section{Introduction}
\label{sec:intro}

Smart contracts are agreements between two or more parties
that are automatically enforced without trusted intermediaries. 
Blockchain technologies reinvented the idea of smart contracts, 
providing trustless environments where they are incarnated as computer programs.
However, writing secure smart contracts is difficult,
as witnessed by the multitude of attacks on smart contracts platforms 
(notably, Ethereum) --- and since smart contracts 
control assets, their bugs may directly lead to financial losses.

Algorand~\cite{Chen19tcs} is a late-generation blockchain that features a
set of interesting features, including high-scalability and a no-forking consensus 
protocol based on Proof-of-Stake~\cite{Alturki19fm}. Its smart contract layer (ASC1) aims to 
mitigate smart contract risks, and adopts a non-Turing-complete programming
model, natively supporting atomic sets of transactions and user-defined assets.
These features make it an intriguing smart contract platform to study.

The official specification and documentation of ASC1 consists of English prose 
and a set of templates to assist programmers in designing their contracts
\cite{AlgorandDeveloperDocs,AlgorandTEALDocs}.
This conforms to standard industry practices,
but there are two drawbacks:
\begin{enumerate}
\item%
  Algorand lacks a mathematical model of contracts and transactions
  suitable for formal reasoning on their behaviour, 
  and for the verification of their properties. 
  Such a model is needed to develop techniques and tools 
  to ensure that contracts are correct and secure;
\item%
  furthermore, %
  even preliminary informal reasoning on non-trivial smart contracts %
  can be challenging,
  as it may require, in some corner cases,
  to resort to experiments, or direct inspection of the platform source code.
\end{enumerate}

\noindent%
Given these drawbacks, we aim at developing a formal model that:
\begin{enumerate}[label={o\arabic*.},ref={o\arabic*}]
\item\label{item:intro-goal1}%
  is high-level enough to simplify the design of Algorand smart contracts
  and  enable formal reasoning about their security properties;
\item\label{item:intro-goal2}%
  expresses Algorand contracts in a simple declarative language,
  similar to PyTeal
  (the official Python binding for Algorand smart contracts)
  \cite{PyTeal};
\item\label{item:intro-goal3}%
  provides a basis for the automatic verification of Algorand smart contracts.
\end{enumerate}

\paragraph{Contributions.} This paper presents:
\begin{itemize}
\item a {\em formal model of stateless ASC1} providing a solid 
theoretical foundation to Algorand smart contracts (\S\ref{sec:model}). Such a model formalises both
  Algorand accounts and transactions
  (\Cref{sec:model:basics}\,--\,\Cref{sec:model:stepL}, \Cref{sec:model:auth}), and 
  smart contracts
  (\Cref{sec:model:contracts});
\item a validation of our model through experiments~\iftoggle{anonymous}{\cite{ASC1AnonymousRepo}}{\cite{ASC1PublicRepo}}
  on the Algorand platform;
\item the formalisation and proof of some {\em fundamental properties of the Algorand state machine}:
  no double spending, determinism, value preservation
  (\Cref{sec:asc1-properties});
\item an {\em analysis of Algorand contract design patterns}
  (\S\ref{sec:smart-contracts}),
  based on several non-trivial contracts
  (covering both standard use cases, and novel ones).
  Quite surprisingly, we show that stateless contracts are expressive enough 
  to encode arbitrary finite state machines;
\item the proof of relevant {\em security properties of smart contracts} in our model;
\item a {\em prototype tool} that compiles smart contracts 
  (written in our formal declarative language)
  into executable TEAL code (\S\ref{sec:tool}).
\end{itemize}

Our formal model is faithful to the actual ASC1 implementation;
by objectives \ref{item:intro-goal1}--\ref{item:intro-goal3},
it strives at being high-level and simple to understand,
while covering the most commonly used primitives and mechanisms of Algorand,
and supporting the specification and verification of non-trivial smart contracts (\S\ref{sec:smart-contracts}, \S\ref{sec:tool}).
To achieve these objectives,
we introduce minor high-level abstractions over low-level details: \eg, since 
TEAL code has the purpose of accepting or rejecting transactions,
we model it using expressions that evaluate to $\true$ or $\false$
(similarly to PyTEAL);
we also formalise different transaction types by focusing on their function, rather than their implementation.
Our objectives imply that we do \emph{not} aim at covering all the possible TEAL contracts with
bytecode-level accuracy, and our Algorand model is \emph{not}
designed as a full low-level formalisation of the behavior of the Algorand
blockchain.
We discuss the differences between our model and the actual Algorand platform in \S\ref{sec:diff}.
\iftoggle{arxiv}
{}
{Due to space constraints, we provide the proofs of our statements in a separate technical report~\cite{abs-2009-12140}.}

%% file: model.tex
\section{The Algorand state machine}
\label{sec:model}

\begin{table}[t]
  \resizebox{\textwidth}{!}{
  \(
  \begin{array}{ll}
    \begin{array}{ll}
      \pmvA, \pmvB, \ldots \mbox{\hspace{35pt}} & \text{Users (key pairs) }
      \\
      \addrX, \addrY, \ldots \in \AddrU & \text{Addresses} 
      \\
      \tokT, \tokTi, \ldots \in \TokU & \text{Assets}
      \\
      \valV, \witW, \ldots \in 0..2^{64}-1 & \text{Values}
      \\
      \tokMapS, \tokMapSi \in \TokU \mapstopart \Nat & \text{Balances}
      \\
      \acct{\addrX}{\tokMapS} & \text{Accounts}
      \\
      \actL, \actLi, \ldots \in \ActU & \text{Transactions}
      \\
      \scriptE, \scriptEi, \ldots & \text{Scripts}
      \\
      \timeT, \timeTi \ldots \in \Nat & \text{Rounds}
    \end{array} & \hspace{0pt}
                  \begin{array}{ll}
                    \actDupSet \subseteq \ActU & \text{Transactions in last $\MaxTxnLife$ rounds}
                    \\
                    \mngrMap \in \TokU \rightarrow \AddrU \mbox{\hspace{0pt}} & \text{Asset manager}
                    \\
                    \leaseMap \in (\AddrU \times \Nat) \mapstopart \Nat \mbox{\hspace{5pt}} & \text{Lease map}
                    \\
                    \freezeMap \in \AddrU \mapstopart \powset{\TokU} \mbox{\hspace{0pt}} & \text{Freeze map}
                    \\
                    \confG, \confGi, \ldots & \text{Blockchain states}
                    \\
                    \acctOk{\tokMapS} & \text{Valid balance}
                    \\
                    \timeOk{\leaseMap,\timeT}{\actL} & \text{Valid time constraint}
                    \\
                    \authver{\witWL}{\actLL,i} & \text{Authorized transaction in group}
                    \\
                    \sem{\actLL,i}{\witWL}{\scriptE} & \text{Script evaluation}
                  \end{array}
  \end{array}
  \)}
  \vspace{5pt}
  \caption{Summary of notation.} 
  \label{fig:model:notation}
\end{table}

We present our formal model of the Algorand blockchain,
including its smart contracts (stateless ASC1), incrementally.
We first define the basic transactions that generate and transfer assets
(\Cref{sec:model:basics}--\Cref{sec:model:step}),
and then add atomic groups of transactions (\Cref{sec:model:stepL}),
smart contracts (\Cref{sec:model:contracts}),
and authorizations (\Cref{sec:model:auth}).
We discuss the main differences between our model and Algorand
in~\Cref{sec:conclusions}.
\iftoggle{anonymous}{\emph{Due to space limits, 
the full formalization is in \Cref{sec:app-model}.}}{}

\subsection{Accounts and transactions}
\label{sec:model:basics}

We use $\pmvA, \pmvB, \ldots$ to denote public/private key pairs
$(\PK{\pmvA},\SK{\pmvA})$.
Users interact with Algorand through pseudonymous identities, 
obtained as a function of their public keys. 
Hereafter, we freely use $\pmvA$ to refer to the public or the private key of $\pmvA$,
or to the user associated with them, relying on the context to resolve the ambiguity.
The purpose of Algorand is to allow users to exchange \keyterm{assets} $\tokT, \tokTi, \ldots$
Besides the Algorand native cryptocurrency $\ALGO$, users can create custom assets.

We adopt the following notational convention:
\begin{itemize}
\item%
  lowercase letters for single entities (\eg, a user $\pmvA$);
\item%
  uppercase letters for \emph{sets} of entities
  (\eg, a set of users $\pmvAS$);
\item%
  calligraphic uppercase letters for \emph{sequences} of entities
  (\eg, list of users $\pmvAL$).
\end{itemize}
Given a sequence $\mathcal{L}$,
we write $|\mathcal{L}|$ for its length,
$\setofseq{\mathcal{L}}$ for the set of its elements,
and $\seqat{\mathcal{L}}{i}$ for its  $i$\textsuperscript{th} element ($i \in 1..|\mathcal{L}|$);
$\emptyseq$ denotes the empty sequence.
We write:
\begin{itemize}
\item%
  $\setenum{\bind{x}{\valV}}$
  for the function mapping $x$ to $\valV$, %
  and having domain equal to $\setenum{x}$;
\item%
  $f \setenum{\bind{x}{\valV}}$
  for the function mapping $x$ to $\valV$, and $y$ to $f(y)$ if $y \neq x$;
\item%
  $f \setenum{\bind{x}{\bot}}$
  for the function undefined at $x$, and mapping $y$ to $f(y)$ if $y \neq x$.
\end{itemize}

\paragraph{Accounts.}
An \keyterm{account} is a deposit of one or more crypto-assets.
We model accounts as terms $\acct{\addrX}{\tokMapS}$, 
where $\addrX$ is an \keyterm{address} uniquely identifying the account, 
and $\tokMapS$ is a \keyterm{balance}, \ie, a finite map from assets to non-negative 64-bit integers.
In the concrete Algorand, %
an address is a 58-characters word; %
for mathematical elegance, in our model
we represent an address as either:
\begin{itemize}
\item a \keyterm{single user} $\pmvA$. Performing transactions on $\acct{\pmvA}{\tokMapS}$ requires $\pmvA$'s authorization;
\item a pair $(\pmvAL,n)$, where $\pmvAL$ is a sequence of users, 
  and $1 \leq n \leq |\pmvAL|$,
  are \keyterm{multisig (multi-signature)} addresses.
  Performing transactions on $\acct{(\pmvAL,n)}{\tokMapS}$ requires that at least $n$ users out of those 
  in $\pmvAL$ grant their authorization;
\footnote{%
  \label{footnote:single-multi-addr}%
    W.l.o.g., we consider a single-user address $\pmvA$
    equivalent to $(\pmvAL,n)$
    with $\pmvAL \!=\! \seq{\pmvA}$, $n\!=\!1$.%
  }%
\item a \keyterm{script}%
\footnote{We formalize scripts (\ie, smart contracts) later on, in~\Cref{sec:model:contracts}.}
$\scriptE$. Performing transactions on $\acct{\scriptE}{\tokMapS}$ requires $\scriptE$ to evaluate to $\true$.
\end{itemize}

Each balance is required to own $\ALGO$s, 
have at least 100000 micro-$\ALGO$s for each owned asset, and cannot control more than 1000 assets.
Formally, we say that $\tokMapS$ is a \keyterm{valid balance} (in symbols, $\acctOk{\tokMapS}$) when:%
\footnote{Since the codomain of $\tokMapS$ is $\Nat$, the balance entry $\tokMapS(\ALGO)$ represents micro-$\ALGO$s.}

\smallskip\centerline{\(
\ALGO \in \dom(\tokMapS)
\; \land \;
\tokMapS(\ALGO) \geq 100000 \cdot |\dom(\tokMapS)|
\; \land \;
|\dom(\tokMapS)| \leq 1001
\)}\smallskip

\paragraph{Transactions.}
Accounts can append various kinds of \keyterm{transactions} to the blockchain,
in order to, \eg, alter their balance or set their usage policies.
We model transactions as records with the structure in~\Cref{fig:model:act}.
Each transaction has a $\actType$, which determines which of the other fields are relevant.%
\footnote{%
  In Algorand, the actual behaviour of a transaction may depend on both its type and other conditions,
  \eg, which optional fields are set.
  For instance, $\actPay$  transactions may also close accounts
  if the \textsf{CloseRemainderTo} field is set.
  For the sake of clarity, in our model we prefer to use a richer set of types;
  see \Cref{sec:conclusions} for other differences.
} %
The field $\actSnd$ usually refers to the subject of the transaction
(\eg, the \emph{sender} in an assets transfer),
while $\actRcv$ refers to the \emph{receiver} in an assets transfer.
The fields $\actTok$ and $\actAmt$ refer, respectively, to the affected asset, and to its amount.
The fields $\actTf$ (``first valid''), $\actTl$ (``last valid'') and $\actLease$ (``lease'')
are used to impose time constraints.

\begin{figure}[t!]
  \centering
    \begin{tabular}{l|l|l}
      \hline
      $\actPay$ & $\actSnd, \actRcv,\actAmt,\actTok$ & $\actSnd$ transfers $\actAmt$ units of $\actTok$ to $\actRcv$ (possibly creating $\actRcv$)
      \\
      $\actClose$ & $\actSnd, \actRcv,\actTok$ & $\actSnd$ gives $\actTok$ to $\actRcv$ and removes it (if $\ALGO$, closes $\actSnd$)
      \\
      $\actGen$ & $\actSnd, \actRcv,\actAmt$ & $\actSnd$ mints $\actAmt$ units of a new asset, managed by $\actRcv$
      \\
      $\actOptin$ & $\actSnd, \actTok$ & $\actSnd$ opts in to receive units of asset $\actTok$
      \\
      $\actBurn$ & $\actTok$ & $\actTok$ is removed from the creator (if sole owner)
      \\
      $\actRevoke$ & $\actSnd, \actRcv,\actAmt,\actTok$ & $\actTok$'s manager transfers $\actAmt$ units of $\actTok$ from $\actSnd$ to $\actRcv$
      \\
      $\actFreeze$ & $\actSnd, \actTok$ & $\actTok$'s manager freezes $\actSnd$'s use of asset $\actTok$
      \\
      $\actUnfreeze$ & $\actSnd, \actTok$ & $\actTok$'s manager unfreezes $\actSnd$'s use of asset $\actTok$
      \\
      $\actDelegate$ & $\actSnd, \actTok$, $\actRcv$ & $\actTok$'s manager delegates $\actTok$ to new manager $\actRcv$
      \\
      \hline
    \end{tabular}
  \vspace{-7pt}
  \caption{Transaction types.
    Fields $\actType$, $\actTf$, $\actTl$, $\actLease$ are common to all types. %
  }
  \label{fig:model:act}
  \vspace{-3mm}%
\end{figure}

Algorand groups transactions into \keyterm{rounds} $\timeT = 1,2,\ldots$
To establish \emph{when} a transaction $\actL$ is valid, we must consider both the current round $\timeT$,
and a \keyterm{lease map} $\leaseMap$ binding pairs (address, lease identifier) to rounds:
this is used to enforce mutual exclusion between two or more transactions
(see \eg the \emph{periodic payment} contract in~\Cref{sec:examples}).
Formally, we define the \keyterm{temporal validity of a transaction} $\actL$ 
by the predicate $\timeOk{\leaseMap,\timeT}{\actL}$, which holds whenever:

\smallskip\centerline{\(%
\begin{array}{l}
  \actL.\actTf \leq \timeT \leq \actL.\actTl
    \quad\text{and}\quad
    \actL.\actTl - \actL.\actTf \leq \MaxTxnLife
    \quad\text{and}
  \\[0mm]
  \big( \actL.\actLease = 0 \quad\text{or}\quad (\actL.\actSnd,\actL.\actLease) \not\in \dom(\leaseMap) \quad\text{or}\quad \timeT > \leaseMap(\actL.\actSnd,\actL.\actLease) \big)
\end{array}
\)}\smallskip%

\noindent%
First, the current round must lie between $\actL.\actTf$ and $\actL.\actTl$,
whose distance cannot exceed $\MaxTxnLife$ rounds%
\footnote{$\MaxTxnLife$ is a global consensus parameter, set to 1000 at time of writing.}.
Second, $\actL$ must have a null lease identifier,
or the identifier has not been seen before (\ie, $\leaseMap(\actL.\actSnd,\actL.\actLease)$ is undefined), 
or the lease has expired (\ie, $\timeT > \leaseMap(\actL.\actSnd,\actL.\actLease)$).
When performed, a transaction with non-null lease identifier 
acquires the lease on $(\actL.\actSnd,\actL.\actLease)$, which is set to $\actL.\actTl$.

\subsection{Blockchain states}
\label{sec:model:blockchain-state}

We model the evolution of the Algorand blockchain as a labelled transition system.
A \keyterm{blockchain state} $\confG$ has the form: 
\begin{equation}
  \label{eq:model:confG}
  \acct{\addrX[1]}{\tokMapS_1} \;\mid\; \cdots \;\mid\; \acct{\addrX[n]}{\tokMapS_n}
  \;\mid\; \timeT \;\mid\; \actDupSet \;\mid\; \mngrMap \;\mid\; \leaseMap \;\mid\; \freezeMap 
\end{equation}
where all addresses $\addrX[i]$ are distinct,
$\mid$ is commutative and associative,
and:
\begin{itemize}
\item $\timeT$ is the current round;
\item $\actDupSet$ is the set of performed transactions whose ``last valid'' time $\actTl$ has not expired. This set is used to avoid double spending (see~\Cref{th:model:no-double-spending});
\item $\mngrMap$ maps each asset to the addresses of its \keyterm{manager} and \keyterm{creator};
\item $\leaseMap$ is the \keyterm{lease map} (from pairs (address, integer) to integers), used to ensure mutual exclusion between transactions;
\item $\freezeMap$ is a map from addresses to sets of assets, used to \keyterm{freeze assets}.
\end{itemize}
We define the \keyterm{initial state} $\confG[0]$ as
$\acct{\addrInit}{\setenum{\bind{\ALGO}{\valV[0]}}} \mid 0 \mid \emptyset \mid \mngrMap \mid \leaseMap \mid \freezeMap$, 
where $\dom(\mngrMap) = \dom(\leaseMap) = \dom(\freezeMap) = \emptyset$,
$\addrInit$ is the initial user address, and $\valV[0] = 10^{16}$
(which is the total supply of $10$ billions $\ALGO$s).%

We now formalize the ASC1 state machine, 
by defining how it evolves by single transactions (\Cref{sec:model:step}), 
and then including atomic groups of transactions (\Cref{sec:model:stepL}),
smart contracts (\Cref{sec:model:contracts}),
and the authorization of transactions (\Cref{sec:model:auth}).

\subsection{Executing single transactions}
\label{sec:model:step}

We write $\confG \step{\actL} \confGi$ to mean: \emph{if} the transaction $\actL$
is performed in blockchain state $\confG$, then the blockchain evolves to state $\confGi$.%
\footnote{\label{footnote:single-trans-auth}%
  Note that $\confG \step{\actL} \confGi$ does not imply that transaction $\actL$ \emph{can} be performed in $\confG$: in fact, $\actL$ might require an authorization. We specify the required conditions in~\Cref{sec:model:auth}.%
}
We specify the transition relation $\step{}$ through a set of inference rules
\iftoggle{arxiv}{%
  (see~\Cref{fig:app-model} in the Appendix for the full definition):
}{%
  (see~\cite[Fig.~5 in Appendix]{abs-2009-12140} for the full definition):
}
each rule describes the effect of a transaction $\actL$ in the state
$\confG$ of~\cref{eq:model:confG}.
We now illustrate all cases,
depending on the transaction type ($\actL.\actType$).

When $\tokT \!\in\! \dom(\tokMapS)$, 
we use the shorthand\, $\tokMapS + \valV{:}\tokT$
to update balance $\tokMapS$
by adding $\valV$ units to token $\tokT$; 
similarly, we write\, $\tokMapS - \valV{:}\tokT$
\;to decrease $\tokT$ by $\valV$ units:

\smallskip\centerline{\(
\tokMapS + \valV{:}\tokT \,\equiv\,
\tokMapS \setenum{\bind{\tokT}{\tokMapS(\tokT) + \valV}}
\qquad\qquad
\tokMapS - \valV{:}\tokT \,\equiv\, 
\tokMapS \setenum{\bind{\tokT}{\tokMapS(\tokT) - \valV}}
\)}

\paragraph{Pay to a new account.}

Let $\actL.\actSnd = \addrX[i]$ for some $i \in 1..n$, 
let $\actL.\actRcv = \addrY \not\in \setenum{\addrX[1],\ldots,\addrX[n]}$
(\ie, the sender account $\addrX$ is already in the state, while the receiver $\addrY$ is not), and let $\actL.\actAmt = \valV$.
The rule has the following preconditions:
\begin{enumerate}[label={{c\arabic*}.},ref={c\arabic*}]
\item \label{item:open:isDup}
  $\actL$ does not cause double-spending ($\actL \not\in \actDupSet$);
\item \label{item:open:timeOk} 
  the time interval of the transaction, and its lease, are respected ($\timeOk{\leaseMap,\timeT}{\actL}$);
\item \label{item:open:acctOkX} 
  the updated balance of $\addrX[i]$ is valid ($\acctOk{\tokMapS[i] - \valV:\ALGO}$);
\item \label{item:open:acctOkY} 
  the balance of the new account at address $\addrY$ is valid ($\acctOk{\setenum{\bind{\ALGO}{\valV}}}$).
\end{enumerate}

\noindent
If these conditions are satisfied, the new state $\confGi$ is the following:

\smallskip\centerline{\(
\acct{\addrX[i]}{\tokMapS[i] - \valV{:}\ALGO} \mid \acct{\addrY}{\setenum{\bind{\ALGO}{\valV}}} \mid \cdots \mid \timeT  \mid \actDupSet \cup \setenum{\actL} \mid \mngrMap \mid \leaseUpdate{\leaseMap}{\actL}{\timeT} \mid \freezeMap
\)}\smallskip

In the new state, the $\ALGO$ balance of $\addrX[i]$ is decreased by $\valV$ units, 
and a new account at $\addrY$ is created, containing exactly the $\valV$ units taken from $\addrX[i]$.
The balances of the other accounts are unchanged.
The updated lease mapping is:

\smallskip\centerline{\(
\leaseUpdate{\leaseMap}{\actL}{\timeT} = \begin{cases}
  \leaseMap \setenum{\bind{(\actL.\actSnd,\actL.\actLease)}{\actL.\actTl}} 
  & \text{if $\actL.\actLease \neq 0$}
  \\
  \leaseMap & \text{otherwise}
\end{cases}
\)}\smallskip

Note that all transaction types check
conditions~\ref{item:open:isDup} and~\ref{item:open:timeOk} above;
further, all transactions check that updated account balances are valid %
(as in \ref{item:open:acctOkX} and \ref{item:open:acctOkY}).%

\paragraph{Pay to an existing account.}

Let $\actL.\actSnd = \addrX[i]$, $\actL.\actRcv = \addrX[j]$, $\actL.\actAmt = \valV$, and $\actL.\actTok = \tokT$.
Besides the common checks,
performing $\actL$ requires that $\addrX[j]$ has ``opted in'' $\tokT$
(formally, $\tokT \in \dom(\tokMapS[j])$),
and $\tokT$ must not be frozen in accounts $\addrX[i]$ and $\addrX[j]$ 
(formally, $\tokT \not\in \freezeMap(\addrX[i]) \cup \freezeMap(\addrX[j])$).
If $\addrX[i] \neq \addrX[j]$, then in the new state the balance of $\tokT$ in $\addrX[i]$ is decreased by $\valV$ units,
and that of $\tokT$ in $\addrX[j]$ is increased by $\valV$ units:

\smallskip\centerline{\(
\acct{\addrX[i]}{\tokMapS[i] - \valV{:}\tokT} 
\mid 
\acct{\addrX[j]}{\tokMapS[j] + \valV{:}\tokT} 
\mid 
\cdots 
\mid 
\timeT 
\mid 
\actDupSet \cup \setenum{\actL} 
\mid 
\mngrMap 
\mid 
\leaseUpdate{\leaseMap}{\actL}{\timeT} 
\mid 
\freezeMap
\)}\smallskip

\noindent%
where all accounts but $\addrX[i]$ and $\addrX[j]$ are unchanged.
Otherwise, if $\addrX[i] = \addrX[j]$, then the balance of $\addrX[i]$ is unchanged, and
the other parts of the state are as above.

\paragraph{Close.}

Let $\actL.\actSnd = \addrX[i]$, $\actL.\actRcv = \addrX[j] \neq \addrX[i]$, and $\actL.\actTok = \tokT$.
Performing $\actL$ has two possible outcomes, 
depending on whether $\tokT$ is $\ALGO$ or a user-defined asset.
If $\tokT = \ALGO$,
we must check that $\tokMapS[i]$ contains \emph{only} $\ALGO$s.
If so, the new state is:

\smallskip\centerline{\(
\acct{\addrX[j]}{\tokMapS[j] + \tokMapS[i](\ALGO){:}\ALGO} 
\mid 
\cdots 
\mid 
\timeT 
\mid 
\actDupSet \cup \setenum{\actL} 
\mid 
\mngrMap 
\mid 
\leaseUpdate{\leaseMap}{\actL}{\timeT} 
\mid 
\freezeMap
\)}\smallskip

\noindent%
where the new state no longer contains the account $\addrX[i]$,
and all the $\ALGO$s in $\addrX[i]$ are transferred to $\addrX[j]$. 
Instead, if $\tokT \neq \ALGO$, performing $\actL$ requires to check only that
$\addrX[i]$ actually contains $\tokT$, and that $\addrX[j]$ has ``opted in'' $\tokT$. 
Further, $\tokT$ must not be frozen for addresses $\addrX[i]$ and $\addrX[j]$, 
\ie $\tokT \not\in \freezeMap(\addrX[i]) \cup \freezeMap(\addrX[j])$.
The new state is:

\smallskip\centerline{\(
\acct{\addrX[i]}{\tokMapS[i]\setenum{\bind{\tokT}{\bot}}}
\mid 
\acct{\addrX[j]}{\tokMapS[j] + \tokMapS[i](\tokT){:}\tokT} 
\mid 
\cdots 
\mid 
\timeT
\mid
\actDupSet \cup \setenum{\actL}
\mid
\mngrMap 
\mid 
\leaseUpdate{\leaseMap}{\actL}{\timeT} 
\mid 
\freezeMap
\)}\smallskip

\noindent%
where $\tokT$ is removed from $\addrX[i]$,
and all the units of $\tokT$ in $\addrX[i]$ are transferred to $\addrX[j]$.

\paragraph{Gen.}

Let $\actL.\actSnd = \addrX[i]$, $\actL.\actRcv = \addrX[j]$, and $\actL.\actAmt = \valV$.
Performing $\actL$ requires that $\addrX[i]$ has enough $\ALGO$s to 
own another asset, \ie $\acctOk{\tokMapS[i] \setenum{\bind{\tokT}{\valV}}}$,
where $\tokT$ is the (fresh) identifier of the new asset.
In the new state, the balance of $\addrX[i]$ is extended with $\setenum{\bind{\tokT}{\valV}}$,
and $\mngrMap$ is updated, making $\addrX[j]$ the manager of  $\tokT$.
The new state is:

\smallskip\centerline{\(
\acct{\addrX[i]}{\tokMapS[i] \setenum{\bind{\tokT}{\valV}}} 
\mid 
\cdots 
\mid 
\timeT
\mid
\actDupSet \cup \setenum{\actL}
\mid 
\mngrMap\setenum{\bind{\tokT}{(\addrX[j],\addrX[i])}} 
\mid 
\leaseUpdate{\leaseMap}{\actL}{\timeT} 
\mid 
\freezeMap
\)}

\paragraph{Opt in.}

Let $\actL.\actSnd = \addrX[i]$ and $\actL.\actTok = \tokT$.
Performing $\actL$ requires that $\tokT$ already occurs in $\confG$,
and that $\addrX[i]$ has enough $\ALGO$s to store it. 
If the balance $\tokMapS[i]$ does not have an entry for $\tokT$,
in the new state $\tokMapS[i]$ is extended with a new entry for $\tokT$:

\smallskip\centerline{\(
\acct{\addrX[i]}{\tokMapS[i] \setenum{\bind{\tokT}{0}}} 
\mid 
\cdots 
\mid 
\timeT
\mid
\actDupSet \cup \setenum{\actL}
\mid 
\mngrMap 
\mid 
\leaseUpdate{\leaseMap}{\actL}{\timeT} 
\mid 
\freezeMap
\)}

\smallskip\noindent
Otherwise, if $\addrX[i]$'s balance has already an entry for $\tokT$, then $\tokMapS[i]$ is unchanged.

\paragraph{Burn.}
Let $\actL.\actSnd = \addrX[i]$ and $\actL.\actTok = \tokT$.
Performing $\actL$ requires that $\addrX[i]$ is the creator of $\tokT$, 
and that $\addrX[i]$ stores \emph{all} the units of $\tokT$
(\ie, there are no units of $\tokT$ in other accounts).
In the resulting state, the token $\tokT$ no longer exists:

\smallskip\centerline{\(
\acct{\addrX[i]}{\tokMapS[i] \setenum{\bind{\tokT}{\bot}}} 
\mid 
\cdots 
\mid 
\timeT
\mid
\actDupSet \cup \setenum{\actL}
\mid 
\mngrMap \setenum{\bind{\tokT}{\bot}}
\mid 
\leaseUpdate{\leaseMap}{\actL}{\timeT} 
\mid 
\freezeMap
\)}

\smallskip\noindent
Note that this transaction requires an authorization by the asset manager of $\tokT$, which is recorded in $\mngrMap$.
(We address this topic in \Cref{sec:model:auth}.)

\paragraph{Revoke.}

Let $\actL.\actSnd = \addrX[i]$ and $\actL.\actRcv = \addrX[j]$. 
Performing $\actL$ requires that both $\addrX[i]$ and $\addrX[j]$ are already storing 
the asset $\tokT$, and that $\tokT$ is not frozen for $\addrX[i]$ and $\addrX[j]$.
In the new state, the balance of $\addrX[i]$ is decreased by $\valV = \actL.\actAmt$ units of 
the asset $\tokT = \actL.\actTok$,
and the balance of $\addrX[j]$ is increased by the same amount:

\smallskip\centerline{\(
\acct{\addrX[i]}{\tokMapS[i] - \valV{:}\tokT} 
\mid 
\acct{\addrX[j]}{\tokMapS[j] + \valV{:}\tokT} 
\mid 
\cdots 
\mid 
\timeT
\mid
\actDupSet \cup \setenum{\actL}
\mid 
\mngrMap 
\mid 
\leaseUpdate{\leaseMap}{\actL}{\timeT} 
\mid 
\freezeMap
\)}\smallskip

\noindent
The effect of a $\actRevoke$ transaction is essentially the same as $\actPay$.
The difference is that $\actRevoke$ must be authorized by the manager of
the asset $\tokT$, while $\actPay$ must be authorized by the sender 
$\addrX[i]$ (see~\Cref{sec:model:auth}).

\paragraph{Freeze and unfreeze.}

A $\actFreeze$ transaction $\actL$ with $\actL.\actSnd = \addrX[i]$ and $\actL.\actTok = \tokT$ updates
the mapping $\freezeMap$ into $\freezeMapi$, such that 
$\freezeMapi(\addrX[i]) = \freezeMap(\addrX[i]) \cup \setenum{\tokT}$,
whenever the asset $\tokT$ is owned by $\addrX[i]$.
This effectively prevents any transfers of the asset $\tokT$ to/from the account $\addrX[i]$. 
The dual transaction $\actUnfreeze$ updates the mapping $\freezeMap$ into $\freezeMapi$ such that 
$\freezeMapi(\addrX[i]) = \freezeMap(\addrX[i]) \setminus \setenum{\tokT}$.

\paragraph{Delegate.}
A $\actDelegate$ transaction $\actL$ with $\actL.\actSnd = \addrX[i]$, $\actL.\actRcv = \addrX[j]$ and $\actL.\actTok = \tokT$ updates the manager of $\tokT$,
provided that $\mngrMap(\tokT) = (\addrX[i],\addrX[i])$,
for some $\addrX[k]$.
In the updated mapping $\mngrMap\setenum{\bind{\tokT}{(\addrX[j],\addrX[k])}}$, 
the manager of $\tokT$ is $\addrX[j]$.

\paragraph{Initiating a new round.}
We model the advancement to the next round of the blockchain as a state transition $\confG \step{\labTime} \confGi$.
In the new state $\confGi$, the round is increased, and the 
set $\actDupSet$ is updated as
\(
\actDupSeti = \setcomp{\actL \in \actDupSet}{\actL.\actTl > \timeT}
\).
The other components of the state are unchanged.



\subsection{Executing atomic groups of transactions}
\label{sec:model:stepL}

\keyterm{Atomic transfers}
allow state transitions to atomically perform \emph{sequences} of transactions.
To atomically perform a sequence $\actLL = \actL[1] \cdots \actL[n]$ from a state $\confG$, 
we must check that all the transactions $\actL[i]$ can be performed \emph{in sequence},
\ie the following precondition must hold (for some $\confG[1], \ldots, \confG[n]$):

\smallskip\centerline{\(
\confG 
\; \step{\actL[1]} \; 
\confG[1]
\quad \cdots \quad
\confG[n-1] \; \step{\actL[n]} \; \confG[n]
\)}\smallskip

\noindent%
If so, the state $\confG$ can take a \emph{single}-step transition labelled $\actLL$. 
Denoting the new transition relation with $\stepL{}$, 
we write the atomic execution of $\actLL$ in $\confG$ as follows:

\smallskip\centerline{\(
\confG
\;\; \stepL{\actLL} \;\;
\confG[n]
\)}

\subsection{Executing smart contracts}
\label{sec:model:contracts}

\begin{figure}[t]
  {\footnotesize
  \begin{align*}
    \scriptE 
    \; \bnfdef \;
    & \valV  
    && \text{constant}
    \\[-1pt]
    \bnfmid 
    & \scriptE \circ \scriptE 
    && \text{arithmetic $(\circ \in \{+, -, <, \leq, =, \geq, >, *, /, \%,\!\andE\!\!,\orE\! \})$}
    \\[-1pt]
    \bnfmid 
    & \notE{\scriptE}
    && \text{negation}
    \\[-1pt]
    \bnfmid
    & \actLen
    && \text{number of transactions in the atomic group}
    \\[-1pt]
    \bnfmid
    & \actPos
    && \text{index of current transaction in the atomic group}
    \\[-1pt]
    \bnfmid
    & \actId[n]
    && \text{identifier of $n$-th transaction in the atomic group}
    \\[-1pt]
    \bnfmid
    & \actE[n]{\actf}
    && \text{value of field $\actf$ of $n$-th transaction in the atomic group}
    \\[-1pt]
    \bnfmid
    & \argE{n}
    && \text{$n$-th argument of the current transaction}
    \\[-1pt]
    \bnfmid
    & \hashE{\scriptE} 
    && \text{hash}
    \\[-1pt]
    \bnfmid
    & \versig{\scriptE}{\scriptE}{\scriptE}
    && \text{signature verification}
  \end{align*}

  \vspace{-2mm}%
  \begin{tabular}{l@{\quad}l}
    Syntactic sugar:&
    $\false \!\bnfdef\! \op{=}{1}{0}$\quad
    $\true \!\bnfdef\! \op{=}{1}{1}$\quad
    $\actE{\actf} \!\bnfdef\! \actE[\actPos]{\actf}$\quad
    $\actId \!\bnfdef\! \actId[\actPos]$
    \\[2pt]
    & $\ifE{\scriptE[0]}{\scriptE[1]}{\scriptE[2]} \bnfdef (\scriptE[0] \andE \scriptE[1]) \orE ((\notE \scriptE[0]) \andE \scriptE[2])$
  \end{tabular}
  }
  \caption{%
    Smart contract scripts %
    (inspired by PyTeal~\cite{PyTeal}).%
  }
  \label{fig:model:scripts}
\end{figure}
\vspace{-1mm}%

In Algorand, custom authorization policies can be defined
with a smart contract language called TEAL \cite{TealSpec}.
TEAL is a bytecode-based stack language,
with an official programming interface for Python (called PyTeal):
in our formal model, we take inspiration from the latter
to abstract TEAL bytecode scripts as terms,
with the syntax in~\Cref{fig:model:scripts}.
Besides standard arithmetic-logical operators,
TEAL includes operators
to count and index all transactions in the current atomic group,
and to access their id and fields.
When firing transaction involving scripts, users can specify a sequence of \keyterm{arguments};
accordingly, the script language includes operators to know the number of arguments, and access them.
Further, scripts include cryptographic operators to compute hashes and verify signatures.

The \keyterm{script evaluation function} $\sem{\actLL,i}{\witWL}{\scriptE}$ 
(\Cref{fig:model:scripts-sem}) evaluates $\scriptE$ using 3 parameters:
a sequence of arguments $\witWL$, a sequence of transactions $\actLL$ forming an atomic group, 
and the index $i \!<\! |\actLL|$ of the transaction containing $\scriptE$.
The script $\actE[n]{\actf}$ evaluates to the field $\actf$ of the $n$\textsuperscript{th} transaction in group $\actLL$.
The size of $\actLL$ is given by $\actLen$, while
$\actPos$ returns the index $i$ of the transaction containing the script being evaluated.
The script $\argE{n}$ returns the $n$\textsuperscript{th} argument in $\witWL$.
The script $\hashE{\scriptE}$ applies a public hash function $\hashSem{}$ to the evaluation of $\scriptE$.
The script $\versig{\scriptE[1]}{\scriptE[2]}{\scriptE[3]}$ verifies a signature $\scriptE[2]$ 
on the message obtained by concatenating the enclosing script and $\scriptE[1]$, 
using public key $\scriptE[3]$.
All operators in~\Cref{fig:model:scripts-sem} are \emph{strict}:
they fail if the evaluation of any operand fails.

\begin{figure*}[t]
  \small
  \resizebox{\textwidth}{!}{
  \(
  \begin{array}{c}
    \sem{\actLL,i}{\witWL}{\valV} 
    = \valV 
    \qquad
    \sem{\actLL,i}{\witWL}{\scriptE \circ \scriptEi} 
    = \sem{\actLL,i}{\witWL}{\scriptE} \circ_\bot \sem{\actLL,i}{\witWL}{\scriptEi}
    \qquad
    \sem{\actLL,i}{\witWL}{\notE \scriptE} 
    = \neg_{\bot} \sem{\actLL,i}{\witWL}{\scriptE}
    \\[5pt]
    \sem{\actLL,i}{\witWL}{\actE[n]{\actf}}
    = (\seqat{\actLL}{n}).\actf
    \;\; (0 \leq n < |\actLL|)
    \qquad
    \sem{\actLL,i}{\witWL}{\actId[n]} = \seqat{\actLL}{n} 
    \;\; (0 \leq n < |\actLL|)
    \\[5pt]
    \sem{\actLL,i}{\witWL}{\actLen} = |\actLL|
    \qquad
    \sem{\actLL,i}{\witWL}{\actPos} = i
    \qquad
    \sem{\actLL,i}{\witWL}{\argE{n}}
    = \seqat{\witWL}{n}
    \;\; (0 \leq n < |\witWL|)
    \\[5pt]
    \sem{\actLL,i}{\witWL}{\hashE{\scriptE}} 
    = \hashSem{\sem{\actLL,i}{\witWL}{\scriptE}}
    \,\quad
    \sem{\actLL,i}{\witWL}{\versig{\scriptE[1]}{\scriptE[2]}{\scriptE[3]}}
    = \versigSem{k}{m}{s}
    \;\;
    \bigg(\begin{array}{l}
      m = \big(\seqat{\actLL}{i}.\actSnd, \sem{\actLL,i}{\witWL}{\scriptE[1]}\big)
      \\[2pt]
      s = \sem{\actLL,i}{\witWL}{\scriptE[2]}
      \;\;
      k = \sem{\actLL,i}{\witWL}{\scriptE[3]} 
    \end{array}
    \bigg)
  \end{array}
  \)
  }
  \vspace{-5pt}
  \caption{Evaluation of scripts in \Cref{fig:model:scripts}.}
  \label{fig:model:scripts-sem}
\end{figure*}

\subsection{Authorizing transactions, and user-blockchain interaction}
\label{sec:model:auth}

As noted before, 
the mere existence of a step $\confG \step{\actL} \confGi$ 
does not imply that $\actL$ can actually be issued.
For this to be possible, 
users must provide a sequence $\witWL$ of \keyterm{witnesses}, 
satisfying the \keyterm{authorization predicate} associated with $\actL$;
such a predicate is uniquely determined by the \keyterm{authorizer address} of $\actL$,
written $\tokMngr[\mngrMap]{\actL}$.
For transaction types $\actClose$, $\actPay$, $\actGen$, $\actOptin$ 
the authorizer address is $\actL.\actSnd$;
for $\actBurn, \actRevoke$, $\actFreeze$ and $\actUnfreeze$ on an asset $\tokT$
it is the asset manager $\mngrMap(\tokT)$.
Intuitively, if $\tokMngr[\mngrMap]{\actL} = \addrX$,
then $\witWL$ authorizes $\actL$~iff:
\begin{enumerate}
\item if $\addrX$ is a multisig address $(\pmvAL,n)$, 
  then $\witWL$ contains at least $n$ signatures of $\actL$,
  made by users in $\pmvAL$;
  (if $\addrX$ is a single-user address $\pmvA$:
  see \cref{footnote:single-multi-addr})
\item if $\addrX$ is a script $\scriptE$, then 
  $\scriptE$ evaluates to $\true$ under the arguments $\witWL$.
\end{enumerate}

We now formalize the intuition above. %
Since the evaluation of scripts depends
on a whole group of transactions $\actLL$,
and on the index $i$ of the current transaction within $\actLL$,
we define the authorization predicate as\; $\authver{\witWL}{\actLL,i}$
\;(read: ``$\witWL$ authorizes the $i$\textsuperscript{th}
transaction in $\actLL$'').
Let $\sig{\pmvAS}{m}$ stand for the set of signatures containing
$\sig{\pmvA}{m}$ for all $\pmvA \in \pmvAS$;
then,\; $\authver{\witWL}{\actLL,i}$ \;holds whenever:
\label{def:tokMngr}
\begin{enumerate}

\item if\; $\tokMngr[\mngrMap]{\seqat{\actLL}{i}} = (\pmvAL,n)$, \;then\; $|\setofseq{\witWL} \cap \sig{\setofseq{\pmvAL}}{\actLL,i} | \geq n$

\item if\; $\tokMngr[\mngrMap]{\seqat{\actLL}{i}} = \scriptE$, \;then\; $\sem{\actLL,i}{\witWL}{\scriptE} = \true$

\end{enumerate}

Note that, in general, the sequence of witnesses $\witWL$ is not unique, 
\ie, it may happen that 
$\authver{\witWL}{\actLL,i}$ and $\authver{\witWLi}{\actLL,i}$
for $\witWL \neq \witWLi$. 
For instance, the $\script{Oracle}$ contract in~\Cref{sec:examples} 
accepts transactions with witnesses of the form $0 \, s$ or $1 \, s'$, 
where the first element of the sequence represents the oracle's choice, 
and the second element is the oracle's signature.

\noindent
Given a \emph{sequence of sequences} of witnesses $\witWLL = \witWL[0] \cdots \witWL[n-1]$ with
$n = |\actLL|$, the \keyterm{group authorization predicate}\,
$\authver{\witWLL}{\actLL}$ \,holds iff\,
$\authver{\witWL[i]}{\actLL,i}$ \,for all $i \!\in\! 0..n-1$.

\paragraph{User-blockchain interaction.}
We model the interaction of users with the blockchain
as a transition system.
Its states are pairs $(\confG,\confK)$, where $\confG$ is a blockchain state, 
while $\confK$ is the set of authorization bitstrings
currently known by users.
The transition relation $\stepTL{\labS}$
(with $\labS \!\in\! \setenum{\witW, \labTime, \auth{\witWLL}{\actLL}}$)
is given by the rules:

\smallskip\centerline{\(
\irule{}
{(\confG,\confK) \stepTL{\witW} (\confG,\confK \cup \setenum{\witW})}
\quad
\irule
{\confG \stepL{\labTime} \confGi}
{(\confG,\confK) \stepTL{\labTime} (\confGi,\confK)}
\quad
\irule
{\confG \stepL{\actLL} \confGi \quad \setofseq{\witWLL} \subseteq \confK \quad \authver{\witWLL}{\actLL}}
{(\confG,\confK) \stepTL{\auth{\witWLL}{\actLL}} (\confGi,\confK)}
\)}\smallskip

With the first two rules, users can broadcast a witness $\witW$,
or advance to the next round.
The last rule gathers from $\confK$ a sequence of witnesses $\witWLL$,
and lets the blockchain perform an atomic group of transactions $\actLL$
if authorized by $\witWLL$.


\subsection{Fundamental properties of ASC1}
\label{sec:asc1-properties}

We now exploit our formal model to establish some fundamental properties of ASC1.
\Cref{th:model:no-double-spending} states that the same transaction $\actL$ cannot be 
issued more than once, \ie, there is no \emph{double-spending}.
In the statement, we use $\stepL{}^* \stepL{\actLL} \,\stepL{}^*$ to denote
an arbitrarily long series of steps including a group of transactions~$\actLL$.

\begin{theorem}[No double-spending]
  \label{th:model:no-double-spending}
  Let $\confG[0] \stepL{}^* \stepL{\actLL} \,\stepL{}^* \confG \stepL{\actLLi} \confGi$.
  Then, no transaction occurs more than once in $\actLL\actLLi$.
\end{theorem}

Define the \emph{value} of an asset $\tokT$ in a state
$\confG = \acct{\addrX[1]}{\tokMapS_1} \mid \cdots \mid \acct{\addrX[n]}{\tokMapS_n} \mid \timeT \mid \cdots$
as the sum of the balances of $\tokT$ in all accounts in $\confG$:

\smallskip\centerline{\(
\tokval{\tokT}{\confG} = \sum_{i = 1}^n \tokval{\tokT}{\tokMapS[i]}
\qquad
\text{where}\;\;
\tokval{\tokT}{\tokMapS} = \begin{cases}
  \tokMapS(\tokT) & \text{if $\tokT \in \dom(\tokMapS)$} \\
  0 & \text{otherwise}
\end{cases}
\)}\smallskip

\Cref{th:model:value-preservation} %
states that, once an asset is minted, 
its value remains constant, until the asset is eventually burnt.
In particular, since $\ALGO$s cannot be burnt (nor minted, unlike in Bitcoin and Ethereum),
their amount remains constant.

\begin{theorem}[Value preservation]
  \label{th:model:value-preservation}
  Let $\confG[0] \stepL{}^* \confG \stepL{}^* \confGi$.
  Then:
  
  \smallskip\centerline{\(%
  \tokval{\tokT}{\confGi} = \begin{cases}
    \tokval{\tokT}{\confG} & \text{if $\tokT$ occurs in $\confG$ and it is not burnt in $\confG \stepL{}^* \confGi$} \\
    0 & \text{otherwise}
  \end{cases}
  \)}%
\end{theorem}

\Cref{th:model:deterministic} establishes that
the transition systems $\stepL{}$ and $\stepTL{}$ are deterministic:
crucially, this allows reconstructing the blockchain state
from the transition log.
Notably, by \cref{th:model:deterministic:stepTL}
of \Cref{th:model:deterministic},
witnesses only determine whether a state transition happens or not,
but they do not affect the new state.
This is unlike Ethereum, where arguments of function calls in transactions
may affect the state.

\begin{theorem}[Determinism]
  \label{th:model:deterministic}
  For all $\lambda \!\in\! \setenum{\labTime, \actLL}$ and $\labS \!\in\! \setenum{\labTime, \witW}$:
  \begin{enumerate}
  \item if $\confG \stepL{\lambda} \confGi$ \,and\, $\confG \stepL{\lambda} \confGii$, 
    \,then\, $\confGi = \confGii$;
  \item if $(\confG,\confK) \stepTL{\labS} (\confGi,\confKi)$ \,and\, 
    $(\confG,\confK) \stepTL{\labS} (\confGii,\confKii)$, \,then\, $(\confGi,\confKi) = (\confGii,\confKii)$;
  \item \label{th:model:deterministic:stepTL}
    if $(\confG,\confK) \stepTL{\!\!\auth{\witWLL}{\actLL}\!} (\confGi,\confKi)$ and 
    $(\confG,\confK) \stepTL{\!\!\auth{\witWLLi}{\actLL}\!} (\confGii,\confKii)$, then 
    $\confGi \!=\! \confGii$ \!and~\mbox{$\confKi \!=\! \confKii \!=\! \confK$.}
  \end{enumerate}
\end{theorem}

%% file: examples.tex
\section{Designing secure smart contracts in Algorand}
\label{sec:examples}

\newcommand{\Escrow}[1][]{\script{Escrow}\ifempty{#1}{}{(#1)}}
\newcommand{\Lottery}[1][]{\script{\it Lottery}\ifempty{#1}{}{(#1)}}
\newcommand{\HTLC}[1][]{\script{\it HTLC}{\ifempty{#1}{}{(#1)}}}

We now exploit our formal model to design some archetypal smart contracts,
and establish their security (\Cref{sec:smart-contracts}).
First, we introduce an attacker model.%

\subsection{Attacker model}
\label{sec:attacker-model}

We assume that cryptographic primitives are secure,
\ie, hashes are collision resistant and signatures cannot be forged
(except with negligible probability).
\label{def:runS}%
A \keyterm{run} $\runS$ is a (possibly infinite) sequence of labels
$\labS[1] \labS[2] \cdots$ %
such that $(\confG[0],\confK[0]) \stepTL{\labS[1]} (\confG[1],\confK[1]) \stepTL{\labS[2]} \cdots$,
where $\confG[0]$ is the initial state, and $\confK[0] = \emptyset$ is the initial (empty) knowledge;
hence, as illustrated in \Cref{sec:model:auth},
each label $\labS[i]$ in a run $\runS$ can be either %
$\witW$ (broadcast of a witness bitstring $\witW$),
$\auth{\witWLL}{\actLL}$ (atomic group of transactions $\actLL$
authorized by $\witWLL$), or
$\labTime$ (advance to next round).
We consider a setting where:
\begin{itemize}
\item%
  each user $\pmvA$ has a \keyterm{strategy} $\stratS{}$, \ie a PPTIME algorithm
  to select which label to perform
  among those permitted by the ASC1 transition system.
  A strategy takes as input a finite run $\runS$ (the past history)
  and outputs a single enabled label $\labS$.
  Strategies are \emph{stateful}:
  users can read and write a private unbounded tape
  to maintain their own state throughout the run.
  The initial state of $\pmvA$'s tape contains $\pmvA$'s private key, 
  and the public keys of all users;%
  \footnote{Notice that new public/private key pairs
    can be generated during the run, 
    and their public parts can be communicated as labels $\witW$.}
\item an \keyterm{adversary} $\Adv$ who controls the scheduling
  with her stateful \keyterm{adversarial strategy} $\stratS{\Adv}$:
  a PPTIME algorithm taking as input the current run $\runS$
  and the labels output by the strategies of users
  (\ie, the steps that users are trying to make).
  The output of $\stratS{\Adv}$ 
  is a single label $\labS$, that is appended to the current run.
  We assume the adversarial strategy $\stratS{\Adv}$
  can delay users' transactions by at most $\delta_{\Adv}$ rounds,
  where $\delta_{\Adv}$ is a given natural number.\footnote{%
    Without this assumption,
    $\Adv$ could arbitrarily disrupt deadlines: 
    \eg, $\stratS{\Adv}$ could make $\pmvA$ \emph{always} lose
    lottery games (like the ones below) by delaying $\pmvA$'s transactions.
  }%
\end{itemize}
A set $\stratSet$ of strategies of users and $\Adv$
induces a distribution of runs;
we say that run $\runS$ is \keyterm{conformant} to $\stratSet$
if $\runS$ is sampled from such a distribution.
We assume that infinite runs contain infinitely many~$\labTime$:
this \emph{non-Zeno condition} ensures that neither users nor $\Adv$ can perform infinitely many transactions in a round.

\subsection{Smart contracts}
\label{sec:smart-contracts}

We now exploit our model to specify some archetypal ASC1 contracts,
and reason about their security.
To simplify the presentation, we assume $\delta_{\Adv} \!=\! 0$,
\ie, the adversary $\Adv$ can start a new round
(performing $\labTime$) only if all users agree.\footnote{%
  All results can be easily adjusted for $\delta_{\Adv} \!>\! 0$,
  but this would require more verbose statements
  to account for possible delays introduced by $\Adv$.%
}
The table below summarises
our selection of smart contracts, highlighting the design patterns
they implement.

\smallskip{%
  \resizebox{0.9\textwidth}{!}{
    \centering
  \begin{tabular}{|c|c|c|c|c|c|c|}
    \hline
    \multirow{2}{*}{\textbf{Use case / Pattern}} & Signed & \multirow{2}{*}{Timeouts} & Commit/  & State & Atomic & Time
    \\
    & witness &  & reveal & machine & transfer & windows
    \\
    \hline
    Oracle & $\checkmark$ & $\checkmark$ & & & &
    \\
    HTLC & & $\checkmark$ & $\checkmark$ & & &
    \\
    Mutual HTLC
    \iftoggle{arxiv}{%
      (\Cref{sec:smart-contracts:mutual-htlc})%
    }{%
      \cite[\S{}B.1]{abs-2009-12140}%
    }%
    & & $\checkmark$ & $\checkmark$ & & $\checkmark$ &
    \\
    $O(n^2)$-collateral lottery & & $\checkmark$ & $\checkmark$ & & $\checkmark$ &
    \\
    0-collateral lottery
    \iftoggle{arxiv}{%
      (\Cref{sec:smart-contracts:zero-collateral-lottery})%
    }{%
      \cite[\S{}B.2]{abs-2009-12140}%
    }%
    & & $\checkmark$ & $\checkmark$ & $\checkmark$ & $\checkmark$ &
    \\
    Periodic payment & & & & & & $\checkmark$
    \\
    Escrow
    \iftoggle{arxiv}{%
      (\Cref{sec:smart-contracts:escrow})%
    }{%
      \cite[\S{}B.3]{abs-2009-12140}%
    }%
    & $\checkmark$ & & & $\checkmark$ & &
    \\
    Two-phase authorization & & $\checkmark$ & & $\checkmark$ & & $\checkmark$
    \\
    Limit order
    \iftoggle{arxiv}{%
      (\Cref{sec:smart-contracts:limit-order})%
    }{%
      \cite[\S{}B.4]{abs-2009-12140}%
    }%
    & & $\checkmark$ & & & $\checkmark$ &
    \\
    Split
    \iftoggle{arxiv}{%
      (\Cref{sec:smart-contracts:split})%
    }{%
      \cite[\S{}B.5]{abs-2009-12140}%
    }%
    & & $\checkmark$ & & & $\checkmark$ &
    \\
    \hline
  \end{tabular}
  } 
}\smallskip

\paragraph{Oracle.}
\label{ex:oracle}

We start by designing a contract
which allows either $\pmvA$ or $\pmvB$ to withdraw all the $\ALGO$s in the contract,
depending on the outcome of a certain boolean event. 
Let $\pmv{o}$ be an oracle who certifies such an outcome, 
by signing the value $1$ or $0$.
We model the contract as the following script:

\smallskip\centerline{\(%
\begin{array}{r@{\;\;}c@{\;\;}l}
  \script{Oracle} &\eqdef&
  \actE{\actType} = \actClose \, \andE \, \actE{\actTok} = \ALGO \, \andE \, \big(
    (\actE{\actTf}>\timeTmax \andE \actE{\actRcv} = \pmvA)
  \\[0mm]
  && \orE\; (\argE{0}=0 \andE \versig{\argE{0}}{\argE{1}}{\pmv{o}} \andE \actE{\actRcv} = \pmvA)
  \\[0mm]
  && \orE\; (\argE{0}=1 \andE \versig{\argE{0}}{\argE{1}}{\pmv{o}} \andE \actE{\actRcv} = \pmvB)
    \big)
\end{array}
\)}\smallskip%

Once created, the contract accepts only $\actClose$ transactions, using two arguments as witnesses.
The argument $\argE{0}$ contains the outcome,
while $\argE{1}$ is $\pmv{o}$'s signature on $(\script{Oracle},\argE{0})$,
\ie, the concatenation between the script and the first argument.
The user $\pmvB$ can collect the funds in $\script{Oracle}$ if $\pmv{o}$ certifies the outcome $1$,
while $\pmvA$ can collect the funds if the outcome is $0$, or after round $\timeTmax$.

\Cref{th:oracle} below proves that $\script{Oracle}$ works as intended.
To state it, we define $\actLS[\pmv{p}]$ as the set of transactions
allowing a user $\pmv{p}$ 
to withdraw the contract funds:

\smallskip\centerline{\(
\actLS[\pmv{p}] = \setcomp{\actL\;}{\;\actL.\actType = \actClose,\; \actL.\actSnd = \script{Oracle},\; \actL.\actRcv = \pmv{p},\; \actL.\actTok = \ALGO}
\)}\smallskip

\noindent%
The \namecref{th:oracle} considers the following strategies for $\pmvA$, $\pmvB$, and $\pmv{o}$:
\begin{itemize}
\item $\stratS{\pmvA}$: wait for $s = \sig{\pmv{o}}{\script{Oracle},0}$; 
  if $s$ arrives at round $\timeT \leq \timeTmax$, 
  then immediately send a transaction $\actL \in \actLS[\pmvA]$ with $\actL.\actTf = \timeT$
  and witness $0 \, s$;
  otherwise, at round $\timeTmax+1$,
  send a transaction $\actL \in \actLS[\pmvA]$ with $\actL.\actTf = \timeTmax+1$;
\item $\stratS{\pmvB}$: wait for $s' = \sig{\pmv{o}}{\script{Oracle},1}$; 
  if $s'$ arrives at round $\timeT$, immediately 
  send a transaction $\actL \in \actLS[\pmvB]$
  with $\actL.\actTf = \timeT$ and witness $1 \, s'$;
\item $\stratS{\pmv{o}}$: do one of the following:
  \begin{enumerate*}[label=\emph{(\alph*)}]
  \item send $\pmv{o}$'s signature on $(\script{Oracle},0)$ at any time, or 
  \item send $\pmv{o}$'s signature on $(\script{Oracle},1)$ at any time, or 
  \item do nothing.
  \end{enumerate*}
\end{itemize}

\begin{theorem}
  \label{th:oracle}
  Let $\runS$ be a run conforming to some set of strategies $\stratSet$, such that:
  \begin{inlinelist}
  \item $\stratS{\pmv{o}} \in \stratSet$;
  \item $\runS$ reaches, at some round before $\timeTmax$, a state $\acct{\script{Oracle}}{\tokMapS} \mid \cdots$;
  \item $\runS$ reaches the round $\timeTmax + 2$.
  \end{inlinelist}
  Then, with overwhelming probability:
  \begin{enumerate}[(1)]
  \item \label{th:oracle:a}
    if $\stratS{\pmvA} \in \stratSet$
    and $\pmv{o}$ has not sent a signature on $(\script{Oracle},1)$, 
    then $\runS$ contains a transaction in $\actLS[\pmvA]$,
    transferring at least $\tokMapS(\ALGO)$ to $\pmvA$;
  \item \label{th:oracle:b}
    if $\stratS{\pmvB} \in \stratSet$ 
    and $\pmv{o}$ has sent a signature on $(\script{Oracle},1)$ at round $\timeT \leq \timeTmax$,
    then $\runS$ contains a transaction in $\actLS[\pmvB]$,
    transferring at least $\tokMapS(\ALGO)$ to $\pmvB$.
  \end{enumerate}
\end{theorem}

Notice that in item~\ref{th:oracle:a} we are only assuming that $\pmvA$ and $\pmv{o}$ use
the strategies $\stratS{\pmvA}$ and $\stratS{\pmv{o}}$, while $\pmvB$ and $\Adv$ can use \emph{any} strategy (and possibly collude).
Similarly, in item~\ref{th:oracle:b} we are only assuming $\pmvB$'s and $\pmv{o}$'s strategies.

\paragraph{Hash Time Lock Contract (HTLC).}
\label{ex:htlc}

A user $\pmvA$ 
promises that she will either reveal a secret $\secret{\pmvA}{}$ by round $\timeTmax$,
or pay a penalty to $\pmvB$.
More sophisticated contracts, \eg gambling games,
use this mechanism to let players generate random numbers in a fair way.
We define the HTLC as the following contract,
parameterised on the two users $\pmvA,\pmvB$ and the hash $\hash{\pmvA}{} = \hashSem{\secret{\pmvA}{}}$ of the secret:

\smallskip\centerline{\(%
\begin{array}{r@{\;\;}c@{\;\;}l@{}l}
  \HTLC[\pmvA,\pmvB,\hash{\pmvA}{}] &\!\eqdef\!&
  \actE{\actType} = \actClose \, \andE \,  \actE{\actTok} = \ALGO \, \andE
  \\[0pt]
  && \hspace{-8pt}\big(
     (\actE{\actRcv} = \pmvA \andE \hashE{\argE{0}} = \hash{\pmvA}{}) \!\orE\!
     (\actE{\actRcv} = \pmvB \andE \actE{\actTf \geq \timeTmax})
     \big)
\end{array}
\)}\smallskip%

\noindent%
The contract accepts only $\actClose$ transactions with receiver $\pmvA$ or $\pmvB$.
User $\pmvA$ can collect the funds in the contract only by providing the secret $\secret{\pmvA}{}$ in $\argE{0}$,
effectively making $\secret{\pmvA}{}$ public.\footnote{%
  If $\secret{\pmvA}{}$ is a sufficiently long bitstring generated uniformly at random, 
  collision resistance of the hash function ensures that only $\pmvA$ (who knows $\secret{\pmvA}{}$) can
  provide such an $\argE{0}$.%
} %
Instead, if $\pmvA$ does not reveal $\secret{\pmvA}{}$, then $\pmvB$ can collect the funds
after round $\timeTmax$.
We state the correctness of $\HTLC$ in \Cref{th:htlc};
first, let $\actLS[\pmv{p}]$ be the set of transactions
allowing user $\pmv{p}$ 
to withdraw the contract funds:

\smallskip\centerline{\(%
\actLS[\pmv{p}] \;=\; \setcomp{\actL\;}{\;\actL.\actType = \actClose,\; \actL.\actSnd = \HTLC[\pmvA,\pmvB,\hash{\pmvA}{}],\; \actL.\actRcv = \pmv{p},\; \actL.\actTok = \ALGO}
\)}\smallskip%

\noindent%
We consider the following strategies for $\pmvA$ and $\pmvB$:
\begin{itemize}
\item $\stratS{\pmvA}$: at a round $\timeT < \timeTmax$, 
  send a $\actL \in \actLS[\pmvA]$ with $\actL.\actTf = \timeT$ 
  and witness~$\secret{\pmvA}{}$;
\item $\stratS{\pmvB}$: at round $\timeTmax$, check whether any transaction in $\actLS[\pmv{a}]$ occurs in $\runS$. 
  If not, then immediately send a transaction $\actL \in \actLS[\pmvB]$ with $\actL.\actTf = \timeTmax$.
\end{itemize}

\begin{theorem}
  \label{th:htlc}
  Let $\runS$ be a run conforming to some set of strategies $\stratSet$, such that:
  \begin{inlinelist}
  \item \label{th:htlc:a}
    $\runS$ reaches, at some round before $\timeTmax$, a state $\acct{\script{\HTLC[\pmvA,\pmvB,\hash{\pmvA}{}]}}{\tokMapS} \mid \cdots$;
  \item \label{th:htlc:b}
    $\runS$ reaches the round $\timeTmax + 1$.
  \end{inlinelist}
  Then, with overwhelming probability: 
  \begin{enumerate}[(1)]
  \item \label{th:htlc:reveal}
    if $\stratS{\pmvA} \in \stratSet$,
    then $\runS$ contains a transaction in $\actLS[\pmvA]$,
    transferring at least $\tokMapS(\ALGO)$ to $\pmvA$;
  \item \label{th:htlc:timeout}
    if $\stratS{\pmvB} \in \stratSet$ 
    and $\runS$ does not contain the secret $\secret{\pmvA}{}$ before round $\timeTmax+1$,
    then $\runS$ contains a transaction in $\actLS[\pmvB]$,
    transferring at least $\tokMapS(\ALGO)$ to $\pmvB$.
  \end{enumerate}
\end{theorem}

\paragraph{Lotteries.}

Consider a gambling game where $n$ players bet $1\ALGO$ each,
and the winner, chosen uniformly at random among them, can redeem $n\,\ALGO$s.
A simple implementation, 
inspired by~\cite{Andrychowicz14sp,Andrychowicz16cacm,bitcoinsok} for Bitcoin, 
requires each player to deposit $n(n-1)\,\ALGO$s as collateral in an HTLC contract.%
\footnote{A zero-collateral lottery is presented in~%
\iftoggle{arxiv}{%
  (\Cref{sec:smart-contracts:zero-collateral-lottery})%
}{%
  \cite[\S{}B.2]{abs-2009-12140}%
}%
.}
For $n=2$ players $\pmvA$ and $\pmvB$, such deposits are transferred by the following transactions:

\smallskip\centerline{\(%
\begin{array}{r@{\;\;}c@{\;\;}l}
  \actL[H\pmvA] 
  &=& \{ \actType:\actPay,\, \actSnd:\pmvA,\, \actRcv:\HTLC[\pmvA,\,\pmvB,\,\hash{\pmvA}{}],\, \actAmt:2,\, \actTok:\ALGO,\, \ldots \}
  \\[0mm]
  \actL[H\pmvB] 
  &=& \{ \actType:\actPay,\, \actSnd:\pmvB,\, \actRcv:\HTLC[\pmvB,\,\pmvA,\,\hash{\pmvB}{}],\, \actAmt:2,\, \actTok:\ALGO,\, \ldots \}
\end{array}
\)}\smallskip%

\noindent
The bets are stored in the following contract,
which determines the winner as a function of the secrets, and allows her to withdraw the whole pot:

\smallskip\centerline{\(%
\begin{array}{r@{\;\;}c@{\;\;}l}
  \Lottery &\!\eqdef\!
  & \actE{\actType} = \actClose \!\andE\! \actE{\actTok} = \ALGO 
    \!\andE\! \hashE{\argE{0}} = \hash{\pmvA}{} 
    \!\andE\! \hashE{\argE{1}} = \hash{\pmvB}{} 
  \\[0mm]
  && \!\andE \ifE{(\argE{0}+\argE{1}) \% 2 = 0}{\actE{\actRcv} = \pmvA}{\actE{\actRcv} = \pmvB}
\end{array}
\)}\smallskip%

\noindent
with $\hash{\pmvA}{} \!\neq\! \hash{\pmvB}{}$.%
\footnote{This check prevents a replay attack: if $\pmvA$ chooses $\hash{\pmvA}{} \!=\! \hash{\pmvB}{}$, then $\pmvB$ cannot win.}
Players $\pmvA$ and $\pmvB$ start the game with the atomic transactions:

\smallskip\centerline{\(%
\begin{array}{r@{\;\;}c@{\;\;}l}
  \multicolumn{3}{c}{
    \actLL[\pmvA,\pmvB] \;=\;
    \actL[H\pmvA] \, \actL[H\pmvB] \, \actL[L\pmvA] \, \actL[L\pmvB]
    \quad\text{where:}
  }\\[0mm]
  \actL[L\pmvA]
  &=& \{ \actType:\actPay,\, \actSnd:\pmvA,\, \actRcv:\Lottery,\, \actAmt:1,\, \actTok:\ALGO,\, \ldots \}
  \\[0mm]
  \actL[L\pmvB]
  &=& \{ \actType:\actPay,\, \actSnd:\pmvA,\, \actRcv:\Lottery,\, \actAmt:1,\, \actTok:\ALGO,\, \ldots \}
\end{array}
\)}\smallskip%

\noindent%
The transaction $\actL[L\pmvA]$ creates the contract with $\pmvA$'s bet, 
and $\actL[L\pmvB]$ completes it with $\pmvB$'s bet.
At this point, there are two possible outcomes:
\begin{enumerate}[label=\emph{(\alph*)}]
\item both players reveal their secret. 
  Then, the winner can withdraw the pot,
  by performing a $\actClose$ action on the $\Lottery$ contract, 
  providing as arguments the two secrets, and setting her identity in the $\actRcv$ field;
\item one of the players does not reveal the secret. 
  Then, the other player can withdraw the collateral in the other player's HTLC (and redeem her own).
\end{enumerate}

To formalise the correctness of the lottery, consider the sets of transactions:

\smallskip\centerline{\(%
\begin{array}{r@{\;\,}c@{\;\,}l}
  \actLSdec[\pmv{p},\pmv{q}]{secr}
  &=& \setcomp{\actL\;}{\;\actL.\actType = \actClose,\, \actL.\actSnd = \HTLC[\pmv{p},\pmv{q},\hash{\pmv{p}}{}],\, \actL.\actRcv = \pmv{p},\, \actL.\actTok = \ALGO}
  \\[2pt]
  \actLSdec[\pmv{p},\pmv{q}]{tout}
  &=& \setcomp{\actL\;}{\;\actL.\actType = \actClose,\, \actL.\actSnd = \HTLC[\pmv{p},\pmv{q},\hash{\pmv{p}}{}],\, \actL.\actRcv = \pmv{q},\, \actL.\actTok = \ALGO}
  \\[2pt]
  \actLSdec[\pmv{p}]{lott}
  &=& \setcomp{\actL\;}{\;\actL.\actType = \actClose,\, \actL.\actSnd = \script{Lottery},\, \actL.\actRcv = \pmv{p},\, \actL.\actTok = \ALGO}
\end{array}
\)}\smallskip%

\noindent
and consider the following strategy $\stratS{\pmvA}$ for $\pmvA$ (the one for $\pmvB$ is analogous):
\begin{enumerate}

\item at some $\timeT < \timeTmax$, 
  send a transaction $\actL \in \actLSdec[\pmvA,\pmvB]{secr}$ with $\actL.\actTf = \timeT$ 
  and witness~$\secret{\pmvA}{}$;

\item if some transaction in $\actLSdec[\pmvB,\pmvA]{secr}$ occurs in $\runS$ 
  at round $\timeTi < \timeTmax$,
  then extract its witness $\secret{\pmvB}{}$ and compute the winner; 
  if $\pmvA$ is the winner, immediately send 
  a transaction $\actL \!\in\! \actLSdec[\pmvA]{lott}$ with $\actL.\actTf = \timeTi$ and
  witness $\secret{\pmvA}{}\secret{\pmvB}{}$;

\item if at round $\timeTmax$ no transaction in $\actLS[\pmvB,\pmvA]$ occurs in $\runS$, 
  immediately send a transaction $\actL \!\in\! \actLSdec[\pmvB,\pmvA]{tout}$ with $\actL.\actTf = \timeTmax$.

\end{enumerate}

\Cref{th:lottery} below establishes that the lottery is fair, 
implying that the expected payoff of player $\pmvA$ following strategy $\stratS{\pmvA}$ 
is at least negligible;
instead, if $\pmvA$ does not follow $\stratS{\pmvA}$ (\eg, by not revealing her secret), the expected payoff may be negative;
analogous results hold for player $\pmvB$.
This result can be generalised for $n \!>\! 2$ players,
with a collateral of $n(n-1)\, \ALGO$s. %
As in the HTLC, we assume that $\secret{\pmvA}{}$ and $\secret{\pmvB}{}$
are sufficiently long bitstrings generated uniformly at random.

\begin{theorem}
  \label{th:lottery}
  Let $\runS$ be a run conforming to a set of strategies $\stratSet$, such that:
  \begin{inlinelist}
  \item $\runS$ contains, before $\timeTmax$, the label $\actLL[\pmvA,\pmvB]$;
  \item $\runS$ reaches round $\timeTmax \!+\! 1$.
  \end{inlinelist}
  For $\pmv{p} \!\neq\! \pmv{q} \!\in\! \setenum{\pmvA,\pmvB}$,
  if $\stratS{\pmv{p}} \!\in\! \stratSet$, then:
  \begin{inlinelist}[(1)]
  \item \label{th:lottery:reveal}
    $\runS$ contains a transaction in $\actLSdec[\pmv{p},\pmv{q}]{secr}$,
    transferring at least $2 \, \ALGO$ to $\pmv{p}$;
  \item \label{th:lottery:timeout}
    the probability that $\runS$ contains
    $\actLSdec[\pmv{q},\pmv{p}]{tout}$ or $\actLSdec[\pmv{p}]{lott}$,
    which transfer at least $1 \, \ALGO$ to $\pmv{p}$,
    is $\geq\!\frac{1}{2}$ (up-to a negligible quantity).%
  \end{inlinelist}
\end{theorem}

\paragraph{Periodic payment.}
We want to ensure that $\pmvA$ can withdraw a fixed amount of $\valV\,\ALGO$s
at fixed time windows of $p$ rounds.
We can implement this behaviour through the following contract,
which can be refilled when needed:

\smallskip\centerline{\(%
\begin{array}{r@{\;\;}c@{\;\;}l}
  \script{PP}(p,d,n)
  &\eqdef
  & \actE{\actType} = \actPay \,\andE\,
    \actE{\actAmt} = \valV \,\andE\, 
    \actE{\actTok} = \ALGO \,\andE\,
  \\[0mm]
  && \actE{\actRcv} = \pmvA \,\andE\,
    \actE{\actTf} \,\%\, p = 0 \,\andE\,
    \actE{\actTl} = \actE{\actTf} + d \,\andE\,
    \actE{\actLease} = n
\end{array}
\)}\smallskip%

\noindent
The contract accepts only $\actPay$ transactions of $\valV\,\ALGO$s to receiver $\pmvA$.
The conditions\; $\actE{\actTf} \,\%\, p = 0$ \;and\; $\actE{\actTl} = \actE{\actTf} + d$ 
\;ensure that the contract only accepts transactions with validity interval
$[k \, p, k \, p + d]$, for $k \!\in\! \Nat$.
The condition\; $\actE{\actLease} = n$ \;ensures that \emph{at most} one such transactions
is accepted for each time window.

\paragraph{Finite-state machines.}

Consider a set of users $\pmvAS$ who want to stipulate a contract whose behaviour is given
by a finite-state machine with states $q_0,\ldots,q_n$.
We can implement such a contract by representing each state $q_i$ as a script $\scriptE[i]$; the current state/script holds the assets,
and each state transition $q_i \rightarrow q_j$
is a clause in $\scriptE[i]$ which enables 
a $\actClose$ transaction to transfer the assets to $\scriptE[j]$.
This clause requires $\actE{\actRcv} = \scriptE[j]$
--- except in case of loops,
which cannot be encoded directly:\footnote{%
  This is because Algorand contracts
  cannot have circular references:
  contract accounts are referenced by script hashes,
  and no script can depend on its own hash.%
}
in this case, we identify the next state as $\actE{\actRcv} = \argE{0}$,
also requiring all users in $\pmvAS$
to sign $\argE{0}$ to confirm its correctness.
To ensure that any user in $\pmvAS$ can trigger a state transition
(by firing the corresponding transaction), their signatures must be exchanged
before the contract starts.
We show an instance of this pattern as the 
two-phase authorization contract below.

An alternative technique is based on quines.
As before, a state transition $q_i \rightarrow q_j$  is rendered as a transaction
which closes $\scriptE[i]$ and transfers the balance to $\scriptE[j]$. 
Here, all such scripts $\scriptE[k]$ have the same code, except for a single state constant $k$ which occurs at a specific offset, and which represents the current state. 
To verify that $\actE{\actRcv}$ represents a legit next state, 
$\scriptE[i]$ requires a witness $\witW$ such that:
\begin{inlinelist}
\item $\actE{\actRcv}$ is equal to the hash of $\witW$, and the state constant $j$ within $\witW$ is indeed a next state for $i$;
\item $\actE{\actSnd}$ is equal to the hash of $\witWi$, where $\witWi$ is obtained from $\witW$ by replacing the state constant $j$ with the current state $i$. 
\end{inlinelist}
Performing these checks could be possible by using concatenation and substring operators.%
\footnote{In Algorand, these operators are available only for $\code{LogicSigVersion} \geq 2$.}

\paragraph{Two-phase authorization.}

We want a contract to allow user $\pmv{c}$ to withdraw some funds,
but only if authorized by $\pmvA$ and $\pmvB$.
We want $\pmvA$ to give her authorization first;
if $\pmvB$'s authorization is not given within $p \geq \MaxTxnLife$ rounds,
then anyone can fire a transaction to reset the contract to its initial state.
We model this contract with two scripts:
$\script{P1}$ represents the state where no authorization has been given yet,
while $\script{P2}$ represents the state where 
$\pmvA$'s authorization has been given.
Conceptually, the contract implements a finite-state machine,
looping between two states until the contract funds are withdrawn by $\pmv{c}$.

\smallskip\centerline{\(%
\begin{array}{r@{\;\;}c@{\;\;}l}
  \script{P1}
  &\eqdef
  & \actE{\actType} = \actClose \,\andE\,
    \actE{\actTok} = \ALGO \,\andE\,
    \versig{\actId}{\argE{0}}{\pmvA} \andE
  \\[0mm]
  && \actE{\actRcv} = \script{P2} \,\andE\,
    \actE{\actTf} \,\%\, (4 * p) = 0 \,\andE\,
    \actE{\actTl} = \actE{\actTf} + \MaxTxnLife
  \\[2pt]
  \script{P2}
  &\eqdef
  & \actE{\actType} = \actClose \,\andE\,
    \actE{\actTok} = \ALGO \,\andE\, 
  \\[0mm]
  && \big( (\versig{\actId}{\argE{0}}{\pmvB} \andE \actE{\actRcv} = \pmv{c}) \orE
  \\[0mm]
  && \;\, (\versig{\argE{0}}{\argE{1}}{\pmvA[1]} \andE \versig{\argE{0}}{\argE{2}}{\pmvB[1]} \andE
  \\[0mm]
  && \hspace{9pt} \actE{\actRcv} = \argE{0} \andE
    \actE{\actTf} \,\%\, (4 * p) = 2 * p \andE 
    \actE{\actTl} = \actE{\actTf} + \MaxTxnLife) \big)
\end{array}
\)}\smallskip%

\noindent%
The scripts $\script{P1}$ and $\script{P2}$ use a time window with $4$ frames,
each lasting $p$ rounds.
Script $\script{P1}$ only accepts $\actClose$ transactions which transfer the balance to $\script{P2}$;
the time constraint ensures that such transactions are sent 
in the first time frame.
The script $\script{P2}$ accepts two kinds of transactions:
\begin{enumerate*}[label=\emph{(\alph*)}]
\item transfer the balance to $\pmv{c}$, using an authorization by $\pmvB$;
\item transfer the balance to $\script{P1}$, in the 4\textsuperscript{th} time frame.
\end{enumerate*}
Note that in $\script{P2}$ we cannot use the (intuitively correct) condition $\actE{\actRcv} = \script{P1}$, as it would introduce a circularity.
Instead, we apply the state machines technique described above:
we require $\actE{\actRcv} = \argE{0}$,
with $\argE{0}$ signed by both $\pmvA$ and $\pmvB$,%
\footnote{We use other key pairs $\pmvA[1]$ and $\pmvB[1]$ to avoid confusion with the signatures on $\actId$.}
and assume that these signatures are exchanged before the contract starts.

%% file: tool.tex
\section{From the formal model to concrete Algorand}
\label{sec:tool}

We now discuss how to translate transactions and scripts in our model to concrete Algorand.
We first sketch how to compile our scripts into TEAL. 
The compilation of most constructs is straightforward.
For instance, a script $\scriptE + \scriptEi$ is compiled by using the opcode \code{+},
and similarly for the other arithmetic and comparison operators, and for the cryptographic primitives.
The logic operators $\!\!\!\andE\!\!$, $\!\!\!\orE\!\!$ 
are compiled via the opcode \code{bnz}, 
to obtain the short-circuit semantics.
The $\notE$ operator is compiled via the opcode \code{!}.
The operator $\actId[n]$ is compiled as \code{gtxn n TxID}, 
$\actLen$ is compiled as \code{global GroupSize}, 
$\actPos$ is compiled as \code{txn GroupIndex}, and
$\argE{n}$ as \code{arg n}.

Finally, compiling the script $\actE[n]{\actf}$ depends on the field $\actf$.
If $\actf$ is $\actTf$, $\actTl$, or $\actLease$, 
then the compilation is \code{gtxn n i}, 
where \code{i} is, respectively, \code{FirstValid}, \code{LastValid}, or \code{Lease}.
For the other cases of $\actf$, the compilation of $\actE[n]{\actf}$
generates a TEAL script which computes $\actf$
by decoding the concrete Algorand transaction fields, 
and making them available in the scratch space.
This decoding is detailed in
\iftoggle{arxiv}{\Cref{fig:app-tool} in~\Cref{sec:app-tool}}{Table 2 in~\cite{abs-2009-12140}}.
From the same table we can also infer how to translate transactions in the model to concrete Algorand transactions.
For instance, translating a transaction of the form:

\smallskip\centerline{\(
\setenum{ \actType : \actClose,\; \actSnd : \addrX,\; \actRcv : \addrY,\; \actTok : \tokT }
\)}\smallskip

\noindent%
results in the concrete transaction in~\Cref{fig:translate-close}
(where we omit the irrelevant fields).

\begin{figure}[t]
\centering
\begin{minipage}{0.4\textwidth}
\begin{lstlisting}[language=asc,numbers=none,numbersep=5pt,xleftmargin=10pt,classoffset=1,morekeywords={},classoffset=2,morekeywords={},framexbottommargin=0pt]
{
    "type": "pay",
    "snd": x,
    "rcv": 0,
    "close": y,
    "amt": 0
}
\end{lstlisting}
\end{minipage}
\begin{minipage}{0.4\textwidth}
\begin{lstlisting}[language=asc,numbers=none,numbersep=5pt,xleftmargin=10pt,classoffset=1,morekeywords={},classoffset=2,morekeywords={},framexbottommargin=0pt]
{
    "type": "axfer",
    "snd": x,
    "asnd": x,
    "arcv": 0,
    "aclose": y,
    "xaid": tau,
    "aamt": 0
}
\end{lstlisting}
\end{minipage}
\vspace{-15pt}
\caption{Translation of a $\actClose$ transaction (left: $\tokT = \ALGO$, right: $\tokT \neq \ALGO$).}
\label{fig:translate-close}
\end{figure}

Our modelling approach is supported by a prototype tool, called 
\secteal (\emph{sec}ure \emph{TEAL}), and accessible via a web 
interface at: 

\smallskip\centerline{%
\url{http://secteal.cs.stir.ac.uk/}
}\smallskip

\noindent
The core of the tool is a compiler that translates smart contracts 
written as expressions, based on the script language (\S\ref{sec:model:contracts}),
into executable TEAL bytecode.
In its current form, \secteal supports experimentation with our model, 
and is provided with a series of examples from \S\ref{sec:smart-contracts}.
Users can also compile their own \secteal contracts, 
paving the way to a declarative approach to contract design and development.
\secteal is a first building block toward a comprehensive IDE for the design, verification,
and deployment of contracts on Algorand.

%% file: conclusions.tex
\section{Conclusions}
\label{sec:conclusions}

This work is part of a wider research line on formal modelling of blockchain-based contracts,
ranging from Bitcoin~\cite{bitcointxm,Klomp18esorics,Rupic20fmbc}
to Ethereum~\cite{Bhargavan16solidether,Luu16ccs,Grishchenko18post,Hirai17wtsc,Hildenbrandt18csf,Grishchenko18cav},
Cardano~\cite{Chakravarty20wtsc},
Tezos~\cite{Bernardo19fm},
and Zilliqa~\cite{Sergey19pacmpl}.
These formal models are a necessary prerequisite to rigorously reason on the security of smart contracts,
and they are the basis for automatic verification.
Besides modelling the behaviour of transactions, in~\Cref{sec:attacker-model} we have proposed
a model of attackers: this enables us to prove properties of smart contracts in the presence
of adversaries, in the spirit of long-standing research in the cryptography area
\cite{Andrychowicz14bw,Andrychowicz14sp,Banasik16esorics,Bentov14crypto,delgado2017fair,Kumaresan14ccs,Kumaresan15ccs}.

\paragraph{Differences between our model and Algorand}
\label{sec:diff}

Besides not modelling the consensus protocol,
to keep the formalization simple, we chose to abstract from some aspects of ASC1,
which do not appear to be relevant to the development of (the majority of) smart contracts.
First, we are not modelling some transaction fields: 
among them, we have omitted the \textsf{fee} field, used to specify an amount of $\ALGO$s to be paid to nodes,
and the \textsf{note} field, used to embed arbitrary bitstrings into transactions.
We associate a single manager to assets, while Algorand uses different managers for different operations
(\eg, the freeze manager for $\actFreeze$/$\actUnfreeze$ and the clawback manager for $\actRevoke$).
We use two different transactions types, $\actPay$ and $\actClose$, 
to perform asset transfers and account closures: in Algorand, a single $\actPay$ transaction 
can perform both.
Note that we can achieve the same effect by performing the $\actPay$ and $\actClose$ transactions
within the same atomic group.
Although Algorand relies on 7 transaction types, the behaviour of some transactions needs to be further qualified by the combination of other fields 
(\eg, freeze and unfreeze are obtained by transactions with the same type \code{afrz}, but with a different value of the \code{AssetFrozen} field).
While this is useful as an implementation detail, our model simplifies reasoning about different behaviours by explicitly exposing them in the transaction type.
In the same spirit, while Algorand uses different transaction types to represent actions with similar functionality (\eg, transferring $\ALGO$s and user-defined assets are rendered with different transaction types, \code{pay} and \code{axfer}), we use the same transaction type (\eg, $\actPay$) for such actions.
Our model does not encompass some advanced features of Algorand, \eg:
rekeying accounts, 
key registration transactions (\code{keyreg}),
some kinds of asset configuration transaction
(\eg, decimals, default frozen, different managers),
and application call transactions.%
\footnote{Application call transactions are used to implement \emph{stateful} contracts, and therefore are outside the scope of this paper.}
Our script language substantially covers TEAL opcodes
with \code{LogicSigVersion=1}, but for a few exceptions,
\eg bitwise operations, different hash functions, jumps.


\paragraph{Related work}

Besides featuring an original consensus protocol 
based on proof-of-stake~\cite{Chen19tcs}, 
Algorand has also introduced a novel paradigm of (stateless) smart contracts,
which differs from the paradigms of other mainstream blockchains.
On the one hand, Algorand follows the \emph{account-based} model,
similarly to Ethereum 
(and differently from Bitcoin and Cardano, which follow the UTXO model).
On the other hand, Algorand's paradigm of stateless contracts 
diverges from Ethereum's stateful contracts:
rather, it resembles Bitcoin’s, where contracts are based upon custom transaction redeem conditions.
Besides these differences, Algorand natively features user-defined assets, 
while other platforms render them as smart contracts 
(\eg, by implementing ERC20 and ERC721 interfaces in Ethereum).
Overall, these differences demand for a formal model 
that is substantially different from 
the models devised for the other blockchain platforms.

Our formalization of the Algorand's script language is close, 
with respect to the level of abstraction, to the model of Bitcoin script 
developed in~\cite{bitcointxm}.
Indeed, both works formalise scripts in an expression language, 
abstracting from the bytecode.
A main difference between Algorand and Ethereum is that Ethereum contracts are stateful: 
their state can be updated by specific bytecode instructions;
instead, (stateless) TEAL scripts merely authorize transactions.
Consequently, a difference between our model and formal models of Ethereum contracts is that the semantics of our scripts has no side effects.
In this way, our work departs from most literature on the formalization of Ethereum
contracts, where the target of the formalization is either the bytecode language EVM~\cite{Luu16ccs,Grishchenko18post,Hildenbrandt18csf},
or the high-level language Solidity~\cite{Crafa19wtsc,BGM19cbt,Jiao20sp}.

\paragraph{Future work}

Our formal model of Algorand smart contracts can be
expanded depending on the evolution of the Algorand framework.
In mid August 2020, Algorand has introduced \emph{stateful} ASC1 contracts
\cite{AlgorandStatefulContracts},
enriching contract accounts with a persistent key-value store,
accessible and modifiable through a new kind of transaction
(which can use an extended set of TEAL opcodes).
To accommodate stateful contracts in our model, we would need to embed 
the key-value store in contract accounts, and extend the
script language with key-value store updates.
The rest of our model (in particular, the semantics of transactions and the attacker model) 
is mostly unaffected by this extension.
Future work could also investigate declarative languages for stateful Algorand smart contracts,
and associated verification techniques.
Another research direction is the mechanization of our formal model, using a proof assistant: 
this would allow machine-checking the proofs developed by pencil-and-paper in~%
\iftoggle{arxiv}{%
  \Cref{sec:proofs}%
}{%
  \cite[\S{}D]{abs-2009-12140}%
}.
Similar work has been done \eg for Bitcoin~\cite{Rupic20fmbc}
and for Tezos~\cite{Bernardo19fm}.

%% file: ack.tex
\paragraph{Acknowledgements}
The authors thank the anonymous reviewers of Financial Cryptography 2021
for their insightful comments on a preliminary version of this paper.
Massimo Bartoletti is partially supported by 
Convenzione tra Fondazione di Sardegna e Atenei Sardi
project F74I19000900007 \emph{ADAM}.
Cristian Lepore is partially supported by The Data Lab, Innovation Center.
Alceste Scalas is partially supported by EU Horizon 2020 project 830929 \textit{CyberSec4Europe}, and Industriens Fonds Cyberprogram 2020-0489 \textit{Security-by-Design in Digital Denmark (Sb3D)}.
Roberto Zunino is partially supported by MIUR PON \textit{Distributed Ledgers for Secure Open Communities}.

%% file: app-model.tex
\section{The stateless ASC1 state machine}
\label{sec:app-model}

We assume the following sets:
\begin{itemize}
\item $\PmvU$, the set of all users;
\item $\ScriptU$, the set of all scripts;
\item $\AddrU = (\PmvU^* \times \Nat) \,\cup\, \ScriptU$, the set of all addresses;
\item $\TokU$, the set of all assets;
\item $\ActU$, the set of all transactions;
\item $\UIntX = \setcomp{\valV \in \Int}{0 \leq \valV < 2^{64}}$, the set of unsigned 64-bit integers.
\end{itemize}

\noindent
We define the partial operator $\circ$ on accounts as follows:
\[
\tokMapS \;\circ\; \valV:\tokT \; = \; 
\tokMapS \setenum{\bind{\tokT}{\tokMapS(\tokT) \circ \valV}} 
\qquad \text{if $\tokT \in \dom(\tokMapS)$ and $\circ \in \setenum{+,-}$}
\]
Note that the use of the $\circ$ operator
in an instance of the rules in~\Cref{fig:app-model} 
on a reachable state $\acct{\addrX}{\tokMapS} \mid \cdots$
never leads to overflows or underflows, 
since these rules correctly keep track of the number of assets.
In particular, overflow never happens since the number of each asset is
bounded ($10^{16}$ for $\ALGO$, and the value $\actAmt$ in $\actGen$
transactions for the other assets).

We define a partial function $\tokMngr[\mngrMap]{\actL}$ which gives the manager of a transaction~$\actL$:
\[
\tokMngr[\mngrMap]{\actL}{} = \begin{cases}
  \actL.\actSnd
  & \text{if } \actL.\actType \in \setenum{\actClose, \actPay, \actGen, \actOptin}
  \\
  \fst(\mngrMap(\actL.\actTok))
  & \text{if } \actL.\actType \in \setenum{\actRevoke,\actFreeze,\actUnfreeze,\actBurn,\actDelegate}
\end{cases}
\]

\noindent
Let $\denV[0], \denV[1] \in \UIntX \cup \setenum{\bot}$.
We define:
\begin{align*}
  \denV[0] \circ_\bot \denV[1] & \equiv \begin{cases}
    \denV[0] \circ \denV[1] & \text{if $\circ \in \setenum{+,-,*}$ and $\denV[0], \denV[1], \denV[0] \circ \denV[1] \in \UIntX$} 
    \\
    \lfloor \denV[0] / \denV[1] \rfloor & \text{if $\circ = /$ and $\denV[0] \in \UIntX$ and $\denV[1] \in \UIntX \setminus \setenum{0}$} 
    \\
    \denV[0] \,\mathrm{mod}\, \denV[1] & \text{if $\circ = \%$ and $\denV[0] \in \UIntX$ and $\denV[1] \in \UIntX \setminus \setenum{0}$}
    \\
    1 & \text{if $\circ \in \setenum{<, \leq, =, \geq, >}$ and $\denV[0],\denV[1] \in \UIntX$ and $\denV[0] \circ \denV[1]$}
    \\
    0 & \text{if $\circ \in \setenum{<, \leq, =, \geq, >}$ and $\denV[0],\denV[1] \in \UIntX$ and $\neg (\denV[0] \circ \denV[1])$}
    \\
    0 & \text{if $\circ = \andE\!$ and $\denV[0] = 0$}
    \\
    \denV[1] & \text{if $\circ = \andE\!$ and $\denV[0] \in \UIntX \setminus \setenum{0}$}
    \\
    1 & \text{if $\circ = \orE\!$ and $\denV[0] \in \UIntX \setminus \setenum{0}$}
    \\
    \denV[1] & \text{if $\circ = \orE\!$ and $\denV[0] = 0$}
    \\
    \bot & \text{otherwise}
  \end{cases}
\end{align*}
\[
\neg_{\bot} \denV[0] = \begin{cases}
  1 & \text{if $\denV[0] = 0$}
  \\
  0 & \text{if $\denV[0] \in \UIntX \setminus \setenum{0}$}
  \\
  \bot & \text{otherwise}
\end{cases}
\]

\noindent
A \emph{blockchain state} $\confG$ is a term with the following syntax:
\[
\confG \; \bnfdef \; 
\acct{\addrX}{\tokMapS} 
\;\;\bnfmid\;\;
\timeT 
\;\;\bnfmid\;\;
\actDupSet
\;\;\bnfmid\;\;
\mngrMap
\;\;\bnfmid\;\;
\leaseMap
\;\;\bnfmid\;\;
\freezeMap
\;\;\bnfmid\;\;
\confG \mid \confGi
\]
and subject to the following conditions:
\begin{itemize}
\item all the terms except $\acct{\addrX}{\tokMapS}$ occur exactly once in a configuration;
\item $\timeT \in \Nat$ is the \emph{current round};
\item $\actDupSet \subseteq \ActU$ is the set of performed transactions whose ``last valid'' time $\actTl$ has not expired;
\item $\mngrMap \in \TokU \rightarrow \AddrU \times \AddrU$ is the \emph{asset map}; 
\item $\leaseMap \in \AddrU \times \Nat \rightarrow  \Nat$ is the \emph{lease map};
\item $\freezeMap \in \AddrU \rightarrow \powset{\TokU}$ is the \emph{freeze map}; 
\item if $\acct{\addrX}{\tokMapS}$ and $\acct{\addrY}{\tokMapSi}$ occur in a configuration, 
  then $\addrX \neq \addrY$.
\item configurations form a commutative monoid with respect to the 
  composition operator $\mid$ (with identity $\confNil$).
\end{itemize}

We define in~\Cref{fig:app-model} a labelled transition relation $\stepL{}$ between blockchain states,
where labels are the following:
\begin{itemize}
\item $\actLL$ performs a sequence of transactions;
\item $\labTime$ advances one round.
\end{itemize}

\begin{figure}[t]
  \resizebox{1.1\textwidth}{!}{%
    \(
    \begin{array}{c}
      \irule
      {\begin{array}{l}
         \actL.\actType=\actPay \quad \actL.\actSnd=\addrX \quad \actL.\actTok = \ALGO \quad \actL.\actRcv=\addrY \quad
         \timeOk{\leaseMap,\timeT}{\actL} \quad
         \actL \not\in \actDupSet
         \\[4pt]
         \acctOk{\tokMapS - \actL.\actAmt:\ALGO} \quad
         \tokMapSi = \setenum{\bind{\ALGO}{\actL.\actAmt}} \quad
         \acctOk{\tokMapSi} \quad
         \acct{\addrY}{\cdots} \text{ not in } \confG
       \end{array}
      }
      {\begin{array}{l}
         \acct{\addrX}{\tokMapS} \mid \confG \mid \timeT \mid \actDupSet \mid \mngrMap \mid \leaseMap  \mid \freezeMap
         \; \step{\actL} \;
         \\[4pt]
         \acct{\addrX}{\tokMapS - \actL.\actAmt:\ALGO} \mid \acct{\addrY}{\tokMapSi} \mid \confG \mid \timeT \mid \actDupSet \cup \setenum{\actL} \mid \mngrMap \mid \leaseUpdate{\leaseMap}{\actL}{\timeT} \mid \freezeMap
       \end{array}
      }
      \; \nrule{[Pay-Open]}
      \\[40pt]
      \irule
      {\begin{array}{l}
         \actL.\actType=\actClose \quad \actL.\actSnd=\addrX \quad \actL.\actRcv=\addrY \quad
         \actL.\actTok=\ALGO \quad 
         \timeOk{\leaseMap,\timeT}{\actL} \quad
         \actL \not\in \actDupSet
         \\[4pt]
         \dom(\tokMapS) = \setenum{\ALGO} \quad
         \acct{\addrY}{\cdots} \text{ not in } \acct{\addrX}{\tokMapS} \mid \confG \quad
       \end{array}
      }
      {\begin{array}{l}
         \acct{\addrX}{\tokMapS} \mid \confG \mid \timeT \mid \actDupSet \mid \mngrMap \mid \leaseMap \mid \freezeMap
         \; \step{\actL} \;
         \\[4pt]
         \acct{\addrY}{\tokMapS} \mid \confG \mid \timeT \mid \actDupSet \cup \setenum{\actL} \mid \mngrMap \mid \leaseUpdate{\leaseMap}{\actL}{\timeT} \mid \freezeMap
       \end{array}
      }
      \; \nrule{[Close-Open]}
      \\[40pt]
      \irule
      {\begin{array}{l}
         \actL.\actType=\actClose \quad \actL.\actSnd=\addrX \quad \actL.\actRcv=\addrY \quad
         \actL.\actTok=\ALGO \quad 
         \timeOk{\leaseMap,\timeT}{\actL} \quad
         \actL \not\in \actDupSet
         \\[4pt]
         \dom(\tokMapS) = \setenum{\ALGO}
       \end{array}
      }
      {\begin{array}{l}
         \acct{\addrX}{\tokMapS} \mid \acct{\addrY}{\tokMapSi} \mid \confG \mid \timeT \mid \actDupSet \mid \mngrMap \mid \leaseMap \mid \freezeMap
         \; \step{\actL} \;
         \\[4pt]
         \acct{\addrY}{\tokMapSi + \tokMapS(\ALGO):\ALGO} \mid \confG \mid \timeT \mid \actDupSet \cup \setenum{\actL} \mid \mngrMap \mid \leaseUpdate{\leaseMap}{\actL}{\timeT} \mid \freezeMap
       \end{array}
      }
      \; \nrule{[Close-Pay]}
      \\[40pt]
      \irule
      {\begin{array}{l}
         \actL.\actType = \actClose \quad
         \actL.\actSnd = \addrX \quad
         \actL.\actRcv = \addrY \quad
         \actL.\actTok = \tokT \quad 
         \timeOk{\leaseMap,\timeT}{\actL} \quad
         \actL \not\in \actDupSet
         \\[4pt]
         \tokT \neq \ALGO \quad
         \tokMapS(\tokT) = \valV \quad
         \tokT \in \dom(\tokMapSi) \quad
         \acctOk{\tokMapSi + \valV:\tokT} \quad
         \tokT \not\in \freezeMap(\addrX) \cup \freezeMap(\addrY)
       \end{array}
      }
      {\begin{array}{l}
         \acct{\addrX}{\tokMapS} \mid \acct{\addrY}{\tokMapSi} \mid \confG \mid \timeT \mid \actDupSet \mid \mngrMap \mid \leaseMap \mid \freezeMap
         \; \step{\actL} \;
         \\[4pt]
         \acct{\addrX}{\tokMapS\setenum{\bind{\tokT}{\bot}}} \mid \acct{\addrY}{\tokMapSi + \valV:\tokT} \mid \confG  \mid \timeT \mid \actDupSet \cup \setenum{\actL} \mid \mngrMap \mid \leaseUpdate{\leaseMap}{\actL}{\timeT} \mid \freezeMap
       \end{array}
      }
      \; \nrule{[Close-Asst]}
      \\[40pt]
      \irule
      {\begin{array}{l}
         \actL.\actType = \actPay \quad
         \actL.\actSnd = \addrX \quad
         \actL.\actRcv = \addrY \quad
         \actL.\actAmt = \valV > 0 \quad
         \actL.\actTok = \tokT \quad 
         \timeOk{\leaseMap,\timeT}{\actL} \quad
         \actL \not\in \actDupSet
         \\[4pt]
         \acctOk{\tokMapS - \valV:\tokT} \quad
         \tokT \in \dom(\tokMapS) \cap \dom(\tokMapSi) \quad
         \tokT \not\in \freezeMap(\addrX) \cup \freezeMap(\addrY)
       \end{array}
      }
      {\begin{array}{l}
         \acct{\addrX}{\tokMapS} \mid \acct{\addrY}{\tokMapSi} \mid \confG \mid \timeT \mid \actDupSet \mid \mngrMap \mid \leaseMap \mid \freezeMap
         \; \step{\actL} \;
         \\[4pt]
         \acct{\addrX}{\tokMapS - \valV:\tokT} \mid \acct{\addrY}{\tokMapSi + \valV:\tokT} \mid \confG \mid \timeT \mid \actDupSet \cup \setenum{\actL} \mid \mngrMap \mid \leaseUpdate{\leaseMap}{\actL}{\timeT} \mid \freezeMap
       \end{array}
      }
      \; \nrule{[Pay]}
      \\[40pt]
      \irule
      {\begin{array}{l}
         \actL.\actType = \actPay \quad
         \actL.\actSnd = \addrX \quad
         \actL.\actRcv = \addrY \neq \addrX \quad
         \actL.\actAmt = 0 \quad
         \timeOk{\leaseMap,\timeT}{\actL} \quad
         \actL \not\in \actDupSet
       \end{array}
      }
      {\begin{array}{l}
         \acct{\addrX}{\tokMapS} \mid \acct{\addrY}{\tokMapSi} \mid \confG \mid \timeT \mid \actDupSet \mid \mngrMap \mid \leaseMap \mid \freezeMap
         \; \step{\actL} \;
         \\[4pt]
         \acct{\addrX}{\tokMapS} \mid \acct{\addrY}{\tokMapSi} \mid \confG \mid \timeT \mid \actDupSet \cup \setenum{\actL} \mid \mngrMap \mid \leaseUpdate{\leaseMap}{\actL}{\timeT} \mid \freezeMap
       \end{array}
      }
      \; \nrule{[Pay-Zero]}
      \\[40pt]
      \irule
      {\begin{array}{l}
         \actL.\actType = \actGen \quad
         \actL.\actSnd = \addrX \quad
         \actL.\actRcv = \addrY \quad
         \actAmt = \valV \quad
         \timeOk{\leaseMap,\timeT}{\actL} \quad
         \actL \not\in \actDupSet
         \\[4pt]
         \acctOk{\tokMapS \setenum{\bind{\tokT}{\valV}}} \quad
         \tokT \text{ next fresh asset identifier}
       \end{array}
      }
      {\begin{array}{l}
         \acct{\addrX}{\tokMapS} \mid \confG \mid \timeT \mid \actDupSet \mid \mngrMap \mid \leaseMap \mid \freezeMap
         \; \step{\actL} \;
         \\[4pt]
         \acct{\addrX}{\tokMapS \setenum{\bind{\tokT}{\valV}}} \mid \confG \mid \timeT \mid \actDupSet \cup \setenum{\actL} \mid \mngrMap\setenum{\bind{\tokT}{(\addrY,\addrX)}} \mid \leaseUpdate{\leaseMap}{\actL}{\timeT} \mid \freezeMap
       \end{array}
      }
      \; \nrule{[Gen]}
      \\[40pt]
      \irule
      {\begin{array}{l}
         \actL.\actType = \actOptin \quad
         \actL.\actSnd = \addrX \quad
         \actL.\actTok = \tokT \quad
         \timeOk{\leaseMap,\timeT}{\actL} \quad
         \actL \not\in \actDupSet
         \\[4pt]
         \tokT \text{ occurs in $\acct{\addrX}{\tokMapS} \mid \confG$} \quad
         \tokMapSi = \begin{cases}
           \tokMapS \setenum{\bind{\tokT}{0}} & \text{if $\tokT \not\in \dom(\tokMapS)$} \\
           \tokMapS & \text{otherwise}
         \end{cases} \quad
                      \acctOk{\tokMapSi}
       \end{array}
                      }
                      {\begin{array}{l}
                         \acct{\addrX}{\tokMapS} \mid \confG \mid \timeT \mid \actDupSet \mid \mngrMap \mid \leaseMap \mid \freezeMap
                         \; \step{\actL} \;
                         \\[4pt]
                         \acct{\addrX}{\tokMapSi} \mid \confG \mid \timeT \mid \actDupSet \cup \setenum{\actL} \mid \mngrMap \mid \leaseUpdate{\leaseMap}{\actL}{\timeT} \mid \freezeMap
                       \end{array}
      }
      \; \nrule{[Optin]}
    \end{array}
    \)
  } 
\end{figure}

\begin{figure}[t]
  \centering
  \(
  \begin{array}{c}
    \irule
    {\begin{array}{l}
       \actL.\actType = \actBurn \quad
       \actL.\actTok = \tokT \quad 
       \snd(\mngrMap(\tokT)) = \addrX \quad
       \timeOk{\leaseMap,\timeT}{\actL} \quad
       \actL \not\in \actDupSet
       \\[4pt]
       \tokT \in \dom(\tokMapS) \quad
       \text{for all $\acct{\addrY}{\tokMapSi}$ in $\confG$}: \tokT \not\in \dom(\tokMapSi)
     \end{array}
    }
    {\begin{array}{l}
       \acct{\addrX}{\tokMapS} \mid \confG \mid \timeT \mid \actDupSet \mid \mngrMap \mid \leaseMap \mid \freezeMap 
       \; \step{\actL} \; 
       \\[4pt]
       \acct{\addrX}{\tokMapS\setenum{\bind{\tokT}{\bot}}} \mid \confG \mid \timeT \mid \actDupSet \cup \setenum{\actL}  \mid \mngrMap\setenum{\bind{\tokT}{\bot}} \mid \leaseUpdate{\leaseMap}{\actL}{\timeT} \mid \freezeMap
     \end{array}
    }
    \; \nrule{[Burn]}
    \\[40pt]
    \irule
    {\begin{array}{l}
       \actL.\actType = \actRevoke \quad
       \actL.\actSnd = \addrX \quad
       \actL.\actRcv = \addrY \quad
       \actL.\actAmt = \valV \quad
       \actL.\actTok = \tokT \quad 
       \timeOk{\leaseMap,\timeT}{\actL} \quad
       \actL \not\in \actDupSet
       \\[4pt]
       \tokT \in \dom(\tokMapS) \cap \dom(\tokMapSi) \quad
       \acctOk{\tokMapS - \valV:\tokT} \quad
       \tokT \not\in \freezeMap(\addrX) \cup \freezeMap(\addrY)
     \end{array}
    }
    {\begin{array}{l}
       \acct{\addrX}{\tokMapS} \mid \acct{\addrY}{\tokMapSi} \mid \confG \mid \timeT \mid \actDupSet \mid \mngrMap \mid \leaseMap \mid \freezeMap 
       \; \step{\actL} \; 
       \\[4pt]
       \acct{\addrX}{\tokMapS - \valV:\tokT} \mid \acct{\addrY}{\tokMapSi + \valV:\tokT} \mid \confG \mid \timeT \mid \actDupSet \cup \setenum{\actL}  \mid \mngrMap \mid \leaseUpdate{\leaseMap}{\actL}{\timeT} \mid \freezeMap
     \end{array}
    }
    \; \nrule{[Revoke]}
    \\[40pt]
    \irule
    {\begin{array}{l}
       \actL.\actType = \actFreeze \quad
       \actL.\actSnd = \addrX \quad
       \actL.\actTok = \tokT \quad 
       \timeOk{\leaseMap,\timeT}{\actL} \quad
       \actL \not\in \actDupSet
       \\[4pt]
       \tokT \in \dom(\tokMapS) \quad
       \freezeMapi = \freezeMap \setenum{\bind{\addrX}{\freezeMap(\addrX) \cup \setenum{\tokT}}}
     \end{array}
    }
    {\begin{array}{l}
       \acct{\addrX}{\tokMapS} \mid \confG \mid \timeT \mid \actDupSet \mid \mngrMap \mid \leaseMap \mid \freezeMap
       \; \step{\actL} \; 
       \\[4pt]
       \acct{\addrX}{\tokMapS} \mid \confG \mid \timeT \mid \actDupSet \cup \setenum{\actL} \mid \mngrMap \mid \leaseUpdate{\leaseMap}{\actL}{\timeT} \mid \freezeMapi
     \end{array}
    }
    \; \nrule{[Freeze]}
    \\[40pt]
    \irule
    {\begin{array}{l}
       \actL.\actType = \actUnfreeze \quad
       \actL.\actSnd = \addrX \quad
       \actL.\actTok = \tokT \quad 
       \timeOk{\leaseMap,\timeT}{\actL} \quad
       \actL \not\in \actDupSet
       \\[4pt]
       \tokT \in \dom(\tokMapS) \quad
       \freezeMapi = \freezeMap \setenum{\bind{\addrX}{\freezeMap(\addrX) \setminus \setenum{\tokT}}}
     \end{array}
    }
    {\begin{array}{l}
       \acct{\addrX}{\tokMapS} \mid \confG \mid \timeT \mid \actDupSet \mid \mngrMap \mid \leaseMap \mid \freezeMap
       \; \step{\actL} \;
       \\[4pt]
       \acct{\addrX}{\tokMapS} \mid \confG \mid \timeT \mid \actDupSet \cup \setenum{\actL} \mid \mngrMap \mid \leaseUpdate{\leaseMap}{\actL}{\timeT} \mid \freezeMapi
     \end{array}
    }
    \; \nrule{[Unfreeze]}
    \\[40pt]
    \irule
    {\begin{array}{l}
       \actL.\actType = \actDelegate \quad
       \actL.\actSnd = \addrX \quad
       \actL.\actRcv = \addrY \quad 
       \actL.\actTok = \tokT \quad 
       \timeOk{\leaseMap,\timeT}{\actL} \quad
       \actL \not\in \actDupSet
       \\[4pt]
       \fst(\mngrMap(\tokT)) = \addrX
     \end{array}
    }
    {\begin{array}{l}
       \confG \mid \timeT \mid \actDupSet \mid \mngrMap \mid \leaseMap \mid \freezeMap
       \; \step{\actL} \;
       \\[4pt]
       \confG \mid \timeT \mid \actDupSet \cup \setenum{\actL} \mid \mngrMap\setenum{\bind{\tokT}{(\addrY,\snd(\mngrMap(\tokT)))}} \mid \leaseUpdate{\leaseMap}{\actL}{\timeT} \mid \freezeMap
     \end{array}
    }
    \; \nrule{[Delegate]}
    \\[40pt]
    \irule
    {}
    {\confG \mid \actDupSet \mid \timeT 
    \; \step{\labTime} \;
    \confG \mid \setcomp{\actL \in \actDupSet}{\actL.\actTl > \timeT} \mid \timeT+1}
    \; \nrule{[Round]}
    \\[20pt]
    \irule
    {
    \actLL = \actL[1] \cdots \actL[n] \quad
    \confG \step{\actL[1]} \confG[1]
    \;\; \cdots \;\;
    \confG[n-1] \step{\actL[n]} \confG[n]
    }
    {\confG
    \stepL{\actLL}
    \confG[n]
    }
    \; \nrule{[TxG]}
    \\[20pt]
    \irule{}
    {(\confG,\confK) \;\stepTL{\witW}\; (\confG,\confK \cup \setenum{\witW})}
    \; \nrule{[Net-Wit]}
    \\[20pt]
    \irule
    {\confG \stepL{\labTime} \confGi}
    {(\confG,\confK) \stepTL{\labTime} (\confGi,\confK)}
    \; \nrule{[Net-Round]}
    \\[20pt]
    \irule
    {\confG \stepL{\actLL} \confGi \quad \setofseq{\witWLL} \subseteq \confK \quad \authver{\witWLL}{\actLL}}
    {(\confG,\confK) \stepTL{\auth{\witWLL}{\actLL}} (\confGi,\confK)}
    \; \nrule{[Net-TxG]}
  \end{array}
  \)
  \caption{The stateless ASC1 state machine.}
  \label{fig:app-model}
\end{figure}

%% file: app-examples.tex
\section{Additional smart contracts}
\label{sec:app-examples}

\subsection{Mutual HTLC}
\label{sec:smart-contracts:mutual-htlc}

The \keyterm{mutual HTLC}~\cite{Andrychowicz14sp} 
is a variant of HTLC presented in~\Cref{sec:smart-contracts}:
two users $\pmvA,\pmvB$ choose their own secrets, pay a deposit,
and the contract ensures that either
\begin{enumerate*}[\emph{(\alph*)}]
\item \emph{both} users reveal their secret and get their deposits back, or
\item whoever does not reveal the secret loses the deposit (in favour of the other user).
\end{enumerate*}

Assuming that $\hash{\pmvA}{}$ and $\hash{\pmvB}{}$ are, respectively, 
the hashes of the secrets $\secret{\pmvA}{}$ and $\secret{\pmvB}{}$,
we can implement the mutual HTLC by creating two HTLC contracts
within an atomic group of transactions:

\smallskip\centerline{\(%
  \{
  \actSnd:\pmvA,\, \actRcv:\HTLC[\pmvA,\pmvB,\hash{\pmvA}{}],\, \ldots \}
  \;\;
  \{ 
  \actSnd:\pmvB,\, \actRcv:\HTLC[\pmvB,\pmvA,\hash{\pmvB}{}],\, \ldots \}
\)}\smallskip%

The mutual HTLC contract guarantees that once $\pmvA$ has stipulated it,  
then she will either learn $\pmvB$'s secret, or receive the compensation.
Instead, if we create two instances of HTLC in a non-atomic way,
this property is not guaranteed, since $\pmvB$ could refuse to stipulate its part of the contract.

\Cref{th:mutual-htlc} states the correctness of the mutual HTLC.
For $\pmv{p},\pmv{q} \in \setenum{\pmvA,\pmvB}$, let:

\smallskip\centerline{\(%
\begin{array}{r@{\;\,}c@{\;\,}l}
  \actLS[\pmv{p},\pmv{q}]
  &=& \setcomp{\actL\;}{\;\actL.\actType = \actClose,\, \actL.\actSnd = \HTLC[\pmv{p},\pmv{q},\hash{\pmv{p}}{}],\, \actL.\actRcv = \pmv{p},\, \actL.\actTok = \ALGO}
  \\
  \actLiS[\pmv{p},\,\pmv{q}]
  &=& \setcomp{\actL\;}{\;\actL.\actType = \actClose,\, \actL.\actSnd = \HTLC[\pmv{p},\pmv{q},\hash{\pmv{p}}{}],\, \actL.\actRcv = \pmv{q},\, \actL.\actTok = \ALGO}
\end{array}
\)}\smallskip%

\noindent%
\alcwarning{Some intuition about the two sets of transactions would help}%
and consider the following strategy for distinct users $\pmv{p} \neq \pmv{q} \in \setenum{\pmvA,\pmvB}$:
\begin{itemize}
\item $\stratS{\pmv{p}}$: at a round $\timeT < \timeTmax$, 
  send a transaction $\actL \in \actLS[\pmv{p},\pmv{q}]$ with $\actL.\actTf = \timeT$ 
  and witness~$\secret{\pmv{p}}{}$.
  Then, at round $\timeTmax$ check whether any transaction in $\actLS[\pmv{q},\pmv{p}]$ occurs in $\runS$:
  if not, immediately send a transaction $\actL \!\in\! \actLiS[\pmv{q},\pmv{p}]$ with $\actL.\actTf = \timeTmax$.
\end{itemize}

\begin{theorem}
  \label{th:mutual-htlc}
  Let $\runS$ be a run conforming to some set of strategies $\stratSet$, such that:
  \begin{inlinelist}
  \item $\runS$ reaches, before $\timeTmax$, a state 
    $\acct{\script{\HTLC[\pmvA,\pmvB,\hash{\pmvA}{}]}}{\tokMapS} \mid \acct{\script{\HTLC[\pmvB,\pmvA,\hash{\pmvB}{}]}}{\tokMapSi} \mid \cdots$;
  \item $\runS$ reaches $\timeTmax + 1$.
  \end{inlinelist}
  Let $\pmv{p} \neq \pmv{q} \in \setenum{\pmvA,\pmvB}$.
  If $\stratS{\pmv{p}} \in \stratSet$, then with overwhelming probability:
  \begin{inlinelist}[(1)]
  \item \label{th:mutual-htlc:reveal}
    $\runS$ contains a transaction in $\actLS[\pmv{p},\pmv{q}]$;
  \item \label{th:mutual-htlc:timeout}
    if $\runS$ does not contain the secret $\secret{\pmv{q}}{}$ before round $\timeTmax+1$,
    then $\runS$ contains a transaction in $\actLiS[\pmv{q},\pmv{p}]$.
  \end{inlinelist}
\end{theorem}

\subsection{Zero-collateral lottery}
\label{sec:smart-contracts:zero-collateral-lottery}

We show a variant of the two-players lottery in~\Cref{sec:smart-contracts}
which requires no collateral, similarly to~\cite{BZ17bw,Miller16zerocollateral}.
The preconditions just require the $1 \ALGO$ bets and the secrets,
while the contract is the following:

\smallskip\centerline{\(%
\begin{array}{r@{\;\;}c@{\;\;}l}
  \script{ZDL} &\eqdef
  & \actE{\actType} = \actClose \, \andE \, \actE{\actTok} = \ALGO \, \andE
  \\
  && \hspace{15pt} 
     \big( (\actE{\actRcv} = \script{ZDL2} \andE \hashE{\argE{0}} = \hash{\pmvA}{}) \orE
  \\[0mm]
  && \hspace{15pt} 
    (\actE{\actRcv} = \pmvB \andE \actE{\actTf} \geq \timeT[0])
    \big)
  \\[4pt]
  \script{ZDL2} &\eqdef
  & \actE{\actType} = \actClose \, \andE \, \actE{\actTok} = \ALGO \, \andE \, \hashE{\argE{0}} = \hash{\pmvA}{} \, \andE \,
  \\
  && \big( ( \hashE{\argE{1}} = \hash{\pmvB}{} \, \andE \, 
  \\
  && \hspace{10pt} \ifE{(\argE{0}+\argE{1}) \% 2 = 0}{\actE{\actRcv} = \pmvA}{\actE{\actRcv} = \pmvB} ) \orE
  \\
  && \hspace{5pt} (\actE{\actRcv} = \pmvA \andE \actE{\actTf} \geq \timeT[0]+\timeT[1])
     \big)
\end{array}
\)}%

\smallskip
Here, $\pmvB$ must reveal first.
If $\pmvB$ does not reveal his secret by the deadline $\timeT[0]$, 
then $\pmvA$ can redeem the $2 \ALGO$s stored in the contract.
Otherwise, $\pmvA$ in turn must reveal by the deadline $\timeT[0]+\timeT[1]$,
or let $\pmvB$ redeem $2 \ALGO$s.
If both $\pmvA$ and $\pmvB$ reveal, then the winner is determined as a function of their secrets.
As before, the rational strategy for each player is to reveal the secret.
This makes the lottery fair, even in the absence of a collateral.

\subsection{Escrow}
\label{sec:smart-contracts:escrow}

User $\pmvA$ wants to buy from seller $\pmvB$ an item that costs $\valV \, \ALGO$s. 
We want to guarantee that
\begin{inlinelist}[\emph{(\arabic*)}]
\item%
  $\pmvB$ will get paid if $\pmvA$ authorizes the payment;
\item%
  $\pmvA$ will be refunded if $\pmvB$ authorizes it;
\item%
  if neither $\pmvA$ nor $\pmvB$ give their authorization, 
  an escrow service $\pmvC$ will resolve the dispute,
  by either fully refunding $\pmvA$, 
  or partially refunding $\pmvA$ with $\valV[\pmvA]\, \ALGO$s 
  and giving the remaining $(\valV-\valV[\pmvA])\, \ALGO$s to $\pmvB$.
\end{inlinelist}

\smallskip\centerline{\(%
\begin{array}{r@{\;\;}c@{\;\;}l}
  \Escrow &\eqdef
  & \actE{\actType} = \actClose \, \andE \, \actE{\actTok} = \ALGO \, \andE \, \big(
  \\[0mm]
  && \qquad \versig{\actId}{\argE{1}}{\pmvA} \, \andE (\actE{\actRcv} = \pmvB \orE \actE{\actRcv} = \script{Resolve})
  \\[0mm]
  && \orE \;\; \versig{\actId}{\argE{1}}{\pmvB} \, \andE (\actE{\actRcv} = \pmvA \orE \actE{\actRcv} = \script{Resolve})
    \big)
  \\[2mm]
  \script{Resolve}  &\eqdef
  & \actE{\actType} = \actPay \, \andE \, \actE{\actTok} = \ALGO \, \andE \, \versig{\argE{0}}{\argE{1}}{\pmvC} \, \andE
  \\[0mm]
  && \big( 
     (\actE{\actRcv} = \pmvA \andE \actE{\actAmt} = \argE{0})
     \orE
  \\[0mm]
  && \hspace{4pt}
     (\actE{\actRcv} = \pmvB \andE \actE{\actAmt} = \valV - \argE{0})
     \big)
\end{array}
\)}%

\subsection{Limit order}
\label{sec:smart-contracts:limit-order}

The \keyterm{limit order} contract~\cite{AlgorandLimitOrder} allows 
a user $\pmvA$ to exchange her $\ALGO$s for units of a certain asset $\tokT$, provided by any user. 
The contract imposes a lower bound $\rho_{\it min}$ to the exchange rate $\tokT$/$\ALGO$,
and guarantees its operation as long as it has enough funds, 
and a deadline $\timeTmax$ is not reached (after then, $\pmvA$ can close it).
To this purpose, the contract accepts two kinds of actions:
\begin{itemize}

\item an atomic group of two $\actPay$ transactions: the first one transfers $\valV[0]$ $\ALGO$s
from the contract to the sender of the second one;
the second transaction transfers $\valV[1]$ units of $\tokT$ to user $\pmvA$.
The contract ensures that:
\begin{inlinelist}
\item the ratio between $\valV[1]$ and $\valV[0]$ is greater then a given constant $\rho_{\it min}$;
\item $\valV[0]$ is greater than a given constant $\valV[{\it min}]$.
\end{inlinelist}
Note that such a group can be issued by \emph{any} user owning enough units of the asset $\tokT$,
without requiring any interaction from $\pmvA$.

\item a single transaction where $\pmvA$ closes the contract after the deadline $\timeTmax$.

\end{itemize}

\noindent
The following script implements this specification:

\medskip\centerline{\(%
\begin{array}{r@{\;\;}c@{\;\;}l}
  &
  & \big(
    \actLen=2 \andE \actPos=0 \andE
  \\[0mm]
  && \hspace{5pt} \actE[0]{\actType} = \actPay \andE \actE[0]{\actTok} = \ALGO \andE \actE[0]{\actRcv} = \actE[1]{\actSnd} \andE
  \\[0mm]
  && \hspace{5pt} \actE[1]{\actType} = \actPay \andE \actE[1]{\actTok} = \tokT \andE \actE[1]{\actRcv} = \pmvA \andE
  \\[0mm]
  && \hspace{5pt} \actE[1]{\actAmt} / \actE[0]{\actAmt} \geq \rho_{\it min} \andE \actE[0]{\actAmt} \geq \valV[\it min]
     \big) 
  \\[0mm]
  && \orE
  \\[0mm]
  && \big( \actLen=1 \andE \actE{\actTf} > \timeTmax \andE
  \\[0mm]
  && \hspace{5pt} \actE{\actType} = \actClose \andE \actE{\actTok} = \ALGO \andE \actE[0]{\actRcv} = \pmvA
  \big)
\end{array}
\)}%

\subsection{Split}
\label{sec:smart-contracts:split}

The \keyterm{split} contract~\cite{AlgorandSplit} is created by a user $\pmvA$,
who want to transfer its funds to users $\pmvB[0]$ and $\pmvB[1]$ in a fixed ratio $\rho$.
The contract is initially funded with some $\ALGO$s from $\pmvA$, and once started it accepts two 
kinds of actions:
\begin{itemize}

\item an atomic group of two $\actPay$ transactions whose sender is the \emph{split} contract: 
the first transaction transfers $\valV[0]$ $\ALGO$s to $\pmvB[0]$,
while the second one transfers $\valV[1]$ $\ALGO$s to $\pmvB[1]$.
The contract ensures that:
\begin{inlinelist}
\item the ratio between $\valV[1]$ and $\valV[0]$ is equal to a given constant $\rho$;
\item $\valV[0]$ is greater than a given constant $\valV[{\it min}]$.
\end{inlinelist}

\item a single transaction where $\pmvA$ closes the contract after a deadline $\timeTmax$.

\end{itemize}

\noindent
The following script implements this specification:

\medskip\centerline{\(%
\begin{array}{r@{\;\;}c@{\;\;}l}
  &
  & \big(
    \actLen=2 \andE \actE[0]{\actSnd} = \actE[1]{\actSnd}  \andE
  \\[0mm]
  && \hspace{5pt} \actE[0]{\actType} = \actPay \andE \actE[0]{\actTok} = \ALGO \andE \actE[0]{\actRcv} = \pmvB[0] \andE
  \\[0mm]
  && \hspace{5pt} \actE[1]{\actType} = \actPay \andE \actE[1]{\actTok} = \ALGO \andE \actE[1]{\actRcv} = \pmvB[1] \andE
  \\[0mm]
  && \hspace{5pt} \actE[1]{\actAmt} = \rho * \actE[0]{\actAmt} \andE \actE[0]{\actAmt} \geq \valV[\it min]
     \big) 
  \\[0mm]
  && \orE
  \\[0mm]
  && \big( \actLen=1 \andE \actE{\actTf} > \timeTmax \andE
  \\[0mm]
  && \hspace{5pt} \actE{\actType} = \actClose \andE \actE{\actTok} = \ALGO \andE \actE[0]{\actRcv} = \pmvA
  \big)
\end{array}
\)}%

%% file: app-tool.tex
\section{Compiling scripts from our model to concrete ASC1}
\label{sec:app-tool}

\Cref{fig:app-tool} summarizes the main checks performed by the script, 
and the values stored in the scratch space.
For brevity, we assume a sugared syntax of TEAL (\eg, we use comparison operators
in infix notation), 
and we do not detail the checks needed to ensure that certain fields 
are not set.
Notice that some of the TEAL operators used in \Cref{fig:app-tool},
like \eg, \code{ConfigAsset}, require $\code{LogicSigVer} \geq 2$
(see \url{https://developer.algorand.org/docs/reference/teal/opcodes/}).
When none of the conditions in the left column of~\Cref{fig:app-tool}
is satisfied, we make the script generate an error,
which corresponds to requiring $\sem{}{}{\actE[n]{\actf}}{}{} = \bot$.

\begin{table}[htbp!]
  \centering
  \resizebox{\textwidth}{!}{
  \begin{tabular}{|l|l|}
    \hline
    \multicolumn{1}{|c|}{\textbf{TEAL checks}} 
    & 
    \multicolumn{1}{|c|}{\textbf{Values stored in scratch space}}
    \\
    \hline
    \begin{tabular}{l}
      \code{gtxn n TypeEnum == pay}
      \\
      \code{gtxn n Amount == 0}
      \\
      \code{gtxn n CloseRemainderTo != 0}
      \\
      $\cdots$
    \end{tabular}
    & 
    \begin{tabular}{l}
      $\actE[n]{\actType} := \actClose$ 
      \\
      $\actE[n]{\actTok} := \ALGO$
      \\
      $\actE[n]{\actSnd} := $ \code{gtxn n Sender}
      \\
      $\actE[n]{\actRcv} := $ \code{gtxn n CloseRemainderTo}
    \end{tabular}
    \\
    \hline
    \begin{tabular}{l}
      \code{gtxn n TypeEnum == pay} 
      \\
      \code{gtxn n CloseRemainderTo == 0}
      \\
      $\cdots$
    \end{tabular}
    & 
    \begin{tabular}{l}
      $\actE[n]{\actType} := \actPay$
      \\
      $\actE[n]{\actTok} := \ALGO$
      \\
      $\actE[n]{\actSnd} := $ \code{gtxn n Sender} 
      \\
      $\actE[n]{\actRcv} := $ \code{gtxn n Receiver}
      \\
      $\actE[n]{\actAmt} := $ \code{gtxn n Amount}
    \end{tabular}
    \\
    \hline
    \begin{tabular}{l}
      \code{gtxn n TypeEnum == axfer} 
      \\
      \code{gtxn n AssetSender != 0}
      \\
      \code{gtxn n AssetAmount == 0}
      \\
      \code{gtxn n AssetCloseTo != 0}
      \\
      $\cdots$
    \end{tabular}
    &
    \begin{tabular}{l}
      $\actE[n]{\actType} := \actClose$
      \\
      $\actE[n]{\actTok} := $ \code{gtxn n XferAsset} 
      \\
      $\actE[n]{\actSnd} := $ \code{gtxn n AssetSender} 
      \\
      $\actE[n]{\actRcv} := $ \code{gtxn n AssetCloseTo}
    \end{tabular}
    \\
    \hline
    \begin{tabular}{l}
      \code{gtxn n TypeEnum == axfer} 
      \\
      \code{gtxn n Sender != 0}
      \\
      \code{gtxn n AssetSender == 0}
      \\
      \code{gtxn n AssetAmount != 0}
      \\
      \code{gtxn n AssetCloseTo == 0}
      \\
      $\cdots$
    \end{tabular}
    &
    \begin{tabular}{l}
      $\actE[n]{\actType} := \actPay$
      \\
      $\actE[n]{\actTok} := $ \code{gtxn n XferAsset} 
      \\
      $\actE[n]{\actSnd} := $ \code{gtxn n Sender} 
      \\
      $\actE[n]{\actRcv} := $ \code{gtxn n AssetReceiver}
      \\
      $\actE[n]{\actAmt} := $ \code{gtxn n AssetAmount}
    \end{tabular}
    \\
    \hline
    \begin{tabular}{l}
      \code{gtxn n TypeEnum == acfg}
      \\
      \code{gtxn n ConfigAsset == 0}
      \\
      \code{(gtxn n ConfigAssetManager ==}
      \\ 
      \code{ gtxn n ConfigAssetFreeze == }
      \\
      \code{ gtxn n ConfigAssetReserve == }
      \\ 
      \code{ gtxn n ConfigAssetClawback) }
      \\
      $\cdots$
    \end{tabular}
    &
    \begin{tabular}{l}
      $\actE[n]{\actType} := \actGen$ 
      \\
      $\actE[n]{\actSnd} := $ \code{gtxn n Sender}
      \\
      $\actE[n]{\actRcv} := $ \code{gtxn n ConfigAssetManager}
      \\
      $\actE[n]{\actAmt} := $ \code{gtxn n AssetAmount}
      \\
    \end{tabular}
    \\
    \hline
    \begin{tabular}{l}
      \code{gtxn n TypeEnum == axfer}
      \\
      \code{gtxn n Sender == gtxn n AssetReceiver}
      \\
      \code{gtxn n AssetSender == 0}
      \\ 
      \code{gtxn n AssetAmount == 0}
      \\
      $\cdots$
    \end{tabular}
    &
    \begin{tabular}{l}
      $\actE[n]{\actType} := \actOptin$ 
      \\
      $\actE[n]{\actTok} := $ \code{gtxn n XferAsset}
      \\
      $\actE[n]{\actSnd} := $ \code{gtxn n Sender} 
    \end{tabular}
    \\
    \hline
    \begin{tabular}{l}
      \code{gtxn n TypeEnum == acfg}
      \\
      \code{len (gtxn n ConfigAssetManager) == 0}
      \\ 
      \code{len (gtxn n ConfigAssetFreeze) == 0}
      \\
      \code{len (gtxn n ConfigAssetReserve) == 0}
      \\ 
      \code{len (gtxn n ConfigAssetClawback) == 0}
      \\
      $\cdots$
      \\
    \end{tabular}
    &
    \begin{tabular}{l}
      $\actE[n]{\actType} := \actBurn$ 
      \\
      $\actE[n]{\actTok} := $ \code{gtxn n ConfigAsset}
    \end{tabular}
    \\
    \hline
    \begin{tabular}{l}
      \code{gtxn n TypeEnum == axfer}
      \\
      \code{gtxn n Sender != 0}
      \\ 
      \code{gtxn n AssetSender != 0}
      \\
      $\cdots$
    \end{tabular}
    &
    \begin{tabular}{l}
      $\actE[n]{\actType} := \actRevoke$ 
      \\
      $\actE[n]{\actTok} := $ \code{gtxn n XferAsset}
      \\
      $\actE[n]{\actSnd} := $ \code{gtxn n AssetSender} 
      \\
      $\actE[n]{\actRcv} := $ \code{gtxn n AssetReceiver} 
      \\
      $\actE[n]{\actAmt} := $ \code{gtxn n AssetAmount}
    \end{tabular}
    \\
    \hline
    \begin{tabular}{l}
      \code{gtxn n TypeEnum == afrz}
      \\
      \code{gtxn n AssetFrozen == true}
      \\
      $\cdots$
    \end{tabular}
    &
    \begin{tabular}{l}
      $\actE[n]{\actType} := \actFreeze$
      \\
      $\actE[n]{\actTok} := $ \code{gtxn n FreezeAsset}
      \\
      $\actE[n]{\actSnd} := $ \code{gtxn n FreezeAccount}
    \end{tabular}
    \\
    \hline
    \begin{tabular}{l}
      \code{gtxn n TypeEnum == afrz}
      \\
      \code{gtxn n AssetFrozen == false}
      \\
      $\cdots$
    \end{tabular}
    &
    \begin{tabular}{l}
      $\actE[n]{\actType} := \actUnfreeze$
      \\
      $\actE[n]{\actTok} := $ \code{gtxn n FreezeAsset}
      \\
      $\actE[n]{\actSnd} := $ \code{gtxn n FreezeAccount}
    \end{tabular}
    \\
    \hline
    \begin{tabular}{l}
      \code{gtxn n TypeEnum == acfg}
      \\
      \code{gtxn n ConfigAsset != 0}
      \\
      \code{len (gtxn n ConfigAssetManager) != 0}
      \\ 
      \code{(gtxn n ConfigAssetManager ==}
      \\ 
      \code{ gtxn n ConfigAssetFreeze == }
      \\
      \code{ gtxn n ConfigAssetReserve == }
      \\ 
      \code{ gtxn n ConfigAssetClawback) }
      \\
      $\cdots$
      \\
    \end{tabular}
    &
    \begin{tabular}{l}
      $\actE[n]{\actType} := \actDelegate$ 
      \\
      $\actE[n]{\actSnd} := $ \code{gtxn n Sender} 
      \\
      $\actE[n]{\actRcv} := $ \code{gtxn n ConfigAssetManager} 
      \\
      $\actE[n]{\actTok} := $ \code{gtxn n ConfigAsset}
      \\
    \end{tabular}
    \\
    \hline
  \end{tabular}
  }
  \caption{Compilation of the script $\actE[n]{\actf}$ in TEAL (sketch).}
  \label{fig:app-tool}
\end{table}

%% file: proofs.tex
\section{Proofs}
\label{sec:proofs}

\begin{proofof}{th:model:no-double-spending}
  By induction on the length of the run.

  Conditions~\ref{item:open:isDup} and~\ref{item:open:timeOk} 
  are required at each step.

  Condition~\ref{item:open:isDup} prevents the same transaction $\actL$
  to appear again for the next $\MaxTxnLife$ rounds
  ($\actL \notin \actDupSet$).
  Condition~\ref{item:open:timeOk} requires $\actL$
  to be valid only in a time window of at most $\MaxTxnLife$ rounds
  ($\actL.\actTl - \actL.\actTf \leq \MaxTxnLife$).
  Therefore, it is impossible for $\actL$ to appear more than once.
  \qed
\end{proofof}

\begin{proofof}{th:model:value-preservation}
  By induction on the length of the run, and by cases on the rules
  generating each step.

  By inspecting the rules, we can see that they all preserve the value
  for all assets, except for $\nrule{[Burn]}$ which completely
  destroys the asset, and $\nrule{[Gen]}$ which can create a fresh asset.  
  The rules prevent a burnt asset to be re-created later on.

  The two special cases are taken into account in the definition of
  $\tokval{\tokT}{\confG}$, so the theorem holds.
  \qed
\end{proofof}

\begin{proofof}{th:model:deterministic}
  By induction on the length of the run, and by cases on the rules
  generating each step.

  Inspecting each pair of (different) rules, we can observe that they
  either have distinct labels or they have conflicting side
  conditions, so that at most one of them can apply in each case.
  Further, we observe that in each rule the final state is a function
  of the label and of the initial state.

  For part 3 of the theorem involving different labels
  $\auth{\witWLL}{\actLL}$ and $\auth{\witWLLi}{\actLL}$, we further
  add that
  ${(\confG,\confK) \stepTL{\auth{\witWLL}{\actLL}} (\confGi,\confK)}$
  is possible only when $\confG \stepL{\actLL} \confGi$ holds, and the
  same holds for $\auth{\witWLLi}{\actLL}$.  We conclude that the
  difference between $\witWLL$ and $\witWLLi$ is immaterial, since the
  determinism of $\stepTL{}$ is derived from the determinism of
  $\stepL{}$ which does not involve witnesses.
  \qed
\end{proofof}

\begin{lemma}
  \label{lem:close}
  Let $\scriptE = (\actE{\actType} = \actClose \andE \scriptEi)$,
  and let $\runS$ be a run passing through a state
  $(\confG,\confK)$ with 
  $\confG = \acct{\scriptE}{\tokMapS} \mid \cdots$
  and leading to $(\confGi,\confKi)$ with 
  $\confGi = \acct{\scriptE}{\tokMapSi} \mid \actDupSeti \mid \timeTi \mid \leaseMapi \cdots$, 
  such that all intermediate states contain the account 
  $\acct{\scriptE}{\cdots}$.
  Then:
  \begin{enumerate}
  \item \label{lem:close:domAlgo}
    $\dom{\tokMapSi} = \setenum{\ALGO}$;
  \item \label{lem:close:monotonic}
    $\tokMapSi(\ALGO) \geq \tokMapS(\ALGO)$;
  \item \label{lem:close:valid}
    if $\actL$ is a transaction such that 
    $\actL.\actSnd = \scriptE$, $\actL.\actTok = \ALGO$, 
    $\actL \not\in \actDupSeti$,
    and $\timeOk{\leaseMapi,\timeTi}{\actL}$,
    and $\sem{\actL}{\witWL}{\scriptE} = \true$ with $\witWL$ in $\confKi$,
    then $(\confGi,\confKi) \stepTL{\auth{\witWL}{\actL}}$.
  \end{enumerate}
\end{lemma}
\begin{proof}
  By construction,
  the script $\scriptE$ only allows $\actClose$ transactions,
  and in particular it does not accept $\actOptin$ or $\actGen$ transactions.
  Therefore, it must be $\dom{\tokMapSi} = \setenum{\ALGO}$,
  proving~\cref{lem:close:domAlgo}.
  Further, the value $\tokMapSi(\ALGO)$ can only increase
  until $\scriptE$ is closed, proving~\cref{lem:close:monotonic}.
  For~\cref{lem:close:valid},
  since by hypothesis $\scriptE$ evaluates to $\true$,
  then it must be $\actL.\actType = \actClose$.
  Then, by definition of $\tokMngr{}$, it must be
  $\tokMngr[\cdots]{\actL} = \scriptE$. 
  Therefore, $\authver{\witWL}{\actL}$ holds,
  since $\sem{\actL}{\witWL}{\scriptE} = \true$.
\end{proof}

Note that, if the script $\scriptE$ does not constrain the $\actLease$ field,
the hypothesis $\actL \not\in \actDupSeti$ in~\cref{lem:close:valid}
can always be satisfied 
by setting a value of $\actL.\actLease$ not used in $\actDupSeti$.
Further, if $\scriptE$ does not check signatures on $\actL$, doing this 
does not invalidate the possibility of appending $\actL$.
Therefore, in the following proofs we will omit checking that
the hypothesis $\actL \not\in \actDupSeti$ holds.

\begin{proofof}{th:oracle}

For item~\ref{th:oracle:a}, 
assume that $\stratS{\pmvA},\stratS{\pmv{o}} \in \stratSet$.
By hypothesis, 
$\pmv{o}$ has not sent a signature $\sig{\pmv{o}}{\script{Oracle},1}$.
By~\Cref{lem:close},
in any reachable state $\acct{\script{Oracle}}{\tokMapSi} \mid \cdots$
it must be $\dom{\tokMapSi} = \setenum{\ALGO}$,
and the value $\tokMapSi(\ALGO)$ can only increase
until $\script{Oracle}$ is closed.
Consequently, to prove item~\ref{th:oracle:a} it enough to show that 
there exists some transaction $\actL$ in $\runS$ such that 
$\actL \in \actLS[\pmvA]$, and 
$\actL$ is the first transaction in $\runS$ whose sender is $\script{Oracle}$.
By the monotonicity of the balance of $\script{Oracle}$ discussed before,
this implies the thesis, 
\ie that $\actL$ transfers at least $\tokMapS(\ALGO)$ to $\pmvA$.
We have the following two cases:
\begin{enumerate}[(a)]

\item $\pmv{o}$ has sent a signature $s$ on $(\script{Oracle},0)$ at round $\timeT \leq \timeTmax$.
By contradiction, assume that $\script{Oracle}$ has not been closed in $\runS$
within round $\timeT$.
Since $\runS$ conforms to $\stratS{\pmvA}$, 
at round $\timeT$, $\pmvA$ has sent a transaction $\actLi \in \actLS[\pmvA]$
with $\actLi.\actTf = \timeT$ and witness $0 \, s$.
Since no transactions have closed $\script{Oracle}$,
the contract is still open, so let $\tokMapSi$ be its balance.
As observed above, $\dom{\tokMapSi} = \setenum{\ALGO}$,
and so $\actLi$ satisfies the premises of either rule 
\nrule{[Close-Open]} or \nrule{[Close-Pay]}.
Further, $\actLi$ satisfies the script $\script{Oracle}$, so 
by~\cref{lem:close:valid} of \Cref{lem:close}
it closes $\script{Oracle}$ at round $\timeT$ --- contradiction.
Therefore, $\runS$ contains at least a transaction 
which closes $\script{Oracle}$ within round $\timeT$: 
let $\actL$ be the first one of these transactions.
Since $\actL$ must satisfy the script $\script{Oracle}$, 
one of the following cases must apply:
\begin{enumerate}[(i)]

\item $\actL.\actTf > \timeTmax$: impossible, because $\actL$ has closed
$\script{Oracle}$ within $\timeT \leq \timeTmax$;

\item $\actL$ is validated by $\pmv{o}$'s signature on $(\script{Oracle},1)$:
impossible, since it would contradict the hypothesis of item~\ref{th:oracle:a};

\item $\actL$ is validated by $\pmv{o}$'s signature on $(\script{Oracle},0)$:
this implies that $\actL \in \actLS[\pmvA]$ --- which proves the thesis.

\end{enumerate}

\item \label{proof:th:oracle:1b}
  $\pmv{o}$ has \emph{not} sent a signature on $(\script{Oracle},0)$ at any round $\timeT \leq \timeTmax$.
By contradiction, assume that $\script{Oracle}$ has not been closed in $\runS$
within round $\timeTmax+1$.
Since $\runS$ conforms to $\stratS{\pmvA}$, at round $\timeTmax+1$, 
$\pmvA$ has sent a transaction $\actLi \in \actLS[\pmvA]$
with $\actLi.\actTf = \timeTmax+1$.
Since no transactions have closed $\script{Oracle}$,
the contract is still open, so let $\tokMapSi$ be its balance.
As observed above, $\dom{\tokMapSi} = \setenum{\ALGO}$,
and so $\actLi$ satisfies the premises of either rule 
\nrule{[Close-Open]} or \nrule{[Close-Pay]}.
Further, $\actLi$ satisfies the script $\script{Oracle}$, so 
by~\cref{lem:close:valid} of \Cref{lem:close}
it closes $\script{Oracle}$ at round $\timeTmax+1$ --- contradiction.
Therefore, $\runS$ contains at least a transaction 
which closes $\script{Oracle}$ within round $\timeTmax+1$: 
let $\actL$ be the first one of these transactions.
Since $\actL$ must satisfy the script $\script{Oracle}$, 
one of the following cases must apply:
\begin{enumerate}[(i)]

\item $\actL.\actTf > \timeTmax$ and $\actL.\actRcv = \pmvA$: 
  this implies that $\actL \in \actLS[\pmvA]$ --- which proves the thesis.

\item $\actL$ is validated by $\pmv{o}$'s signature on $(\script{Oracle},1)$:
impossible, since it would contradict the hypothesis of item~\ref{th:oracle:a};

\item $\actL$ is validated by $\pmv{o}$'s signature on $(\script{Oracle},0)$:
impossible, since it would contradict the hypothesis 
of case~\ref{proof:th:oracle:1b} of item~\ref{th:oracle:a}.

\end{enumerate}

\end{enumerate}

\medskip\noindent
For item~\ref{th:oracle:b}, assume that
$\stratS{\pmvB},\stratS{\pmv{o}} \in \stratSet$.
By~\Cref{lem:close},
in any reachable state $\acct{\script{Oracle}}{\tokMapSi} \mid \cdots$
it must be $\dom{\tokMapSi} = \setenum{\ALGO}$,
and the value $\tokMapSi(\ALGO)$ can only increase
until $\script{Oracle}$ is closed.
Consequently, to prove item~\ref{th:oracle:b} it enough to show that 
there exists some transaction $\actL$ in $\runS$ such that 
$\actL \in \actLS[\pmvB]$, and 
$\actL$ is the first transaction in $\runS$ whose sender is $\script{Oracle}$.
By the monotonicity of the balance of $\script{Oracle}$ discussed before,
this implies the thesis, 
\ie that $\actL$ transfers at least $\tokMapS(\ALGO)$ to $\pmvB$.

By hypothesis, $\pmv{o}$ has sent a signature $s'$ on
$(\script{Oracle},1)$ at round $\timeT \leq \timeTmax$.
By contradiction, assume that $\script{Oracle}$ has not been closed in $\runS$
within round $\timeT$.
Since $\runS$ conforms to $\stratS{\pmvB}$, 
at round $\timeT$, $\pmvB$ has sent a transaction $\actLi \in \actLS[\pmvB]$
with $\actLi.\actTf = \timeT$ and witness $1 \, s'$.
Since no transactions have closed $\script{Oracle}$,
the contract is still open, so let $\tokMapSi$ be its balance.
As observed above, $\dom{\tokMapSi} = \setenum{\ALGO}$,
and so $\actLi$ satisfies the premises of either rule 
\nrule{[Close-Open]} or \nrule{[Close-Pay]}.
Further, $\actLi$ satisfies the script $\script{Oracle}$, so 
by~\cref{lem:close:valid} of \Cref{lem:close}
it closes $\script{Oracle}$ at round $\timeT$ --- contradiction.
Therefore, $\runS$ contains at least a transaction 
which closes $\script{Oracle}$ within round $\timeT$: 
let $\actL$ be the first one of these transactions.
Since $\actL$ must satisfy the script $\script{Oracle}$, 
one of the following cases must apply:
\begin{enumerate}[(i)]

\item $\actL.\actTf > \timeTmax$: impossible, because
  $\actL$ has closed $\script{Oracle}$ within round $\timeT \leq \timeTmax$;

\item $\actL$ is validated by $\pmv{o}$'s signature on $(\script{Oracle},1)$:
  this implies that $\actL \in \actLS[\pmvB]$ 
  --- which proves the thesis.

\item $\actL$ is validated by $\pmv{o}$'s signature on
  $(\script{Oracle},0)$: impossible, since $\runS$ would not conform to 
  $\stratS{\pmv{o}}$, which requires that the oracle never signs
  both $0$ and $1$.
  \qed

\end{enumerate}
\end{proofof}

\begin{proofof}{th:htlc}

For item~\ref{th:htlc:a}, assume that
$\stratS{\pmvA} \in \stratSet$.
By contradiction, assume that $\runS$ does not contain a transaction
in $\actLS[\pmvA]$.
Since $\runS$ conforms to $\stratS{\pmvA}$, at round $\timeT$, user
$\pmvA$ has sent a transaction $\actL \in \actLS[\pmvA]$ with
$\actL.\actTf = \timeT$ and witness $\secret{\pmvA}{}$.
By construction, $\actL$ satisfies the script $\script{HTLC}$.
Hence, by~\Cref{lem:close}, 
to prove the thesis we show that $\script{HTLC}$ was not
closed earlier in $\runS$.
By inspecting the script $\script{HTLC}$, there are only two
conditions under which a previous transaction $\actLi$ can close the
contract:
\begin{enumerate}

\item $\actLi.\actTf > \timeTmax$: impossible, because
  $\actL.\actTf = \timeT \leq \timeTmax$;

\item $\actLi$ is validated using a preimage of $\hash{\pmvA}{}$: this
  implies that $\actLi \in \actLS[\pmvA]$ --- contradiction.

\end{enumerate}

\noindent
For item~\ref{th:htlc:b}, assume that
$\stratS{\pmvB} \in \stratSet$, and that $\secret{\pmvA}{}$
was not revealed before round $\timeTmax+1$.
By contradiction, assume that $\runS$ does not contain a transaction
in $\actLS[\pmvB]$.
Since $\runS$ conforms to $\stratS{\pmvB}$, at round $\timeTmax$, user
$\pmvB$ has sent a transaction $\actL \in \actLS[\pmvB]$ with
$\actL.\actTf = \timeTmax$.
By construction, $\actL$ satisfies the script $\script{HTLC}$.
Hence, by~\Cref{lem:close},
to prove the thesis we show that $\script{HTLC}$ was not
closed earlier in $\runS$.
By inspecting the script $\script{HTLC}$, there are only two
conditions under which a previous transaction $\actLi$ can close the
contract:

\begin{enumerate}

\item $\actLi.\actTf > \timeTmax$: this implies that
  $\actLi \in \actLS[\pmvB]$ --- contradiction.

\item $\actLi$ is validated using a preimage of $\hash{\pmvA}{}$: this
  implies, with overwhelming probability, that $\secret{\pmvA}{}$ was
  used as the preimage. In such case, it was revealed and occurs in
  $\runS$ at round $\timeTmax$ --- contradiction.

\end{enumerate}

\end{proofof}

\begin{proofof}{th:lottery}

\noindent
The proof for the fairness of the lottery protocol when implemented on
Bitcoin appeared in~\cite{Andrychowicz14sp,Andrychowicz16cacm}.
\zunnote{check cite, ma penso sia OK}%
This protocol relies on running, simultaneously, two timed commitments
(for which we use $\HTLC$) and a contract that transfers the bets to
the winner, which can be computed after the secrets have been revealed
(for which we use $\Lottery$).

We already proved the security of $\HTLC$ in~\Cref{th:htlc}, which
implies part 1 
\zunnote{la label non funziona}%
of the thesis: a honest participant $\pmv{p}$ will reveal her secret
in time with a transaction in $\actLSdec[\pmv{p},\pmv{q}]{secr}$.

For part 2, 
the argument is similar to the one for the original protocol.
Briefly put, if the other participant $\pmv{q}$ does
not to reveal $\secret{\pmv{q}}{}$ in time,
$\pmv{p}$ wins the lottery by timeout using a transaction
in $\actLSdec[\pmv{q},\pmv{p}]{tout}$.
If instead $\secret{\pmv{q}}{}$ is revealed in time, since
$\secret{\pmv{p}}{}$ was chosen in a uniformly random way,
independently from $\secret{\pmv{q}}{}$ (which has a different hash),
we have that the parity of $\secret{\pmv{p}}{} + \secret{\pmv{q}}{}$
is a random bit, uniformly distributed.
Hence, if $\pmv{p}$ won, her strategy makes her send a transaction in
$\actLSdec[\pmv{p}]{lott}$ and claim the pot.

\noindent%
In the argument above we neglected the case where $\pmv{q}$ reveals
another preimage than $\secret{\pmv{q}}{}$ as his secret, since that
can only happen with negligible probability.

\noindent%
Summing up, an honest $\pmv{p}$ wins with at least $1/2$ probability
(up-to a  negligible quantity).
\end{proofof}

\begin{proofof}{th:mutual-htlc}

  The proof is analogous to the one for HTLC in~\Cref{th:htlc}.
  
\end{proofof}
